\documentclass[journal]{IEEEtran}

\usepackage{cite}
\ifCLASSINFOpdf

\else

\fi

\usepackage{bm}
\usepackage{amsmath}
\usepackage{epsf}
\usepackage{graphics}
\usepackage{ amssymb }
\usepackage[dvips]{graphicx}
\usepackage{epsfig}
\usepackage{cite}
\usepackage[linesnumbered,ruled,vlined]{algorithm2e}
\usepackage{colortbl}
\usepackage{color}

\usepackage{bm}
\usepackage{amsmath}
\usepackage{epsf}
\usepackage{graphics}
\usepackage{ amssymb }
\usepackage[dvips]{graphicx}
\usepackage{epsfig}
\usepackage{cite}
\usepackage[linesnumbered,ruled,vlined]{algorithm2e}
\usepackage{graphicx}
\usepackage{epsfig}
\usepackage{latexsym}
\usepackage{amsfonts}
\usepackage{here}
\usepackage{rawfonts}
\usepackage{hyperref}

\usepackage[utf8]{inputenc}
\usepackage[english]{babel}
\usepackage{amsmath}
\usepackage{amsfonts}
\usepackage{amssymb}
\usepackage{color} %May be necessary if you want to color links
\usepackage{bm}
\usepackage{listings}
\usepackage{caption}
\usepackage{amssymb}
\usepackage{amsthm}
\usepackage{graphicx}
\usepackage{epstopdf}
\usepackage{listings}
\usepackage{float}
\usepackage{amsmath}
\usepackage{amssymb}
\usepackage{amsfonts}
\usepackage{epstopdf}
\usepackage[overload]{empheq}

\usepackage{multirow}
\usepackage{amscd}
\usepackage{mathrsfs}
\usepackage{graphicx}
\usepackage{color}
\usepackage{url}
\usepackage{bm}
\usepackage{setspace}

\usepackage{footnote}
\DeclareMathOperator*{\argmax}{arg\,max}

\newcommand*{\Scale}[2][4]{\scalebox{#1}{$#2$}}%
\newtheorem{theorem}{Theorem}

\newtheorem{definition}{Definition}

\newtheorem{corollary}{Corollary}
\newtheorem{example}{Example}
\newtheorem{remark}{Remark}

\usepackage[]{algorithm2e}
\newcommand\floor[1]{\lfloor#1\rfloor}
\newcommand\ceil[1]{\lceil#1\rceil}

\newcommand\norm[1]{\left\lVert#1\right\rVert}
\usepackage{mathtools}
\DeclarePairedDelimiter\abs{\lvert}{\rvert}
\usepackage{lipsum,graphicx,subcaption}
\usepackage[protrusion=true,expansion=true]{microtype}
\usepackage[font=small]{caption}
\begin{document}
	
	\title{Infection Analysis on Irregular Networks through Graph Signal Processing}

	\author{Seyyedali Hosseinalipour,~\IEEEmembership{Student member,~IEEE,}
    Jie Wang,~\IEEEmembership{Student member,~IEEE,}
		Yuanzhe Tian,
		Huaiyu Dai,~\IEEEmembership{Fellow,~IEEE}
		\thanks{Seyyedali Hosseinalipour, Jie Wang and Huaiyu Dai are with the Department
			of Electrical and Computer Engineering, North Carolina State University, Raleigh,
			NC, USA (e-mail: shossei3@ncsu.edu; jwang50@ncsu.edu; hdai@ncsu.edu).
            Yuanzhe Tian is with the school of Information and Electronics, Beijing Institute of Technology, China (e-mail: 1120141466@bit.edu.cn).}
            \thanks{Part of this work was presented at the IEEE Global Communications Conference (GLOBECOM) 2017 \cite{ref:ours}.}}


	\maketitle
	
	\begin{abstract}
In a networked system, functionality can be seriously endangered when nodes are \textit{infected}, due to internal random failures or a contagious virus that develops into an epidemic. Given a \textit{snapshot} of the network representing the nodes' states (infected or healthy), \textit{infection analysis} refers to distinguishing an epidemic from random failures and gathering information for effective countermeasure design. This analysis is challenging due to irregular network structure, heterogeneous epidemic spreading, and noisy observations. This paper treats a network snapshot as a \textit{graph signal}, and develops effective approaches for infection analysis based on graph signal processing. For the macro (network-level) analysis aiming to distinguish an epidemic from random failures, 1) multiple detection metrics are defined based on the graph Fourier transform (GFT) and neighborhood characteristics of the graph signal; 2) a new class of graph wavelets, distance-based graph wavelets (DBGWs), are developed; and 3) a machine learning-based framework is designed employing either the GFT spectrum or the graph wavelet coefficients as features for infection analysis. DBGWs also enable the micro (node-level) infection analysis, through which the performance of epidemic countermeasures can be improved. Extensive simulations are conducted to demonstrate the effectiveness of all the proposed algorithms in various network settings.
	\end{abstract}
	

	\IEEEpeerreviewmaketitle

	\section{Introduction}
    
    \IEEEPARstart{N}{owadays}, the size of the networked systems, such as mobile wireless networks, computer networks, cloud networks and smart grids, keeps growing. These large-scale networks are increasingly prone to a wide variety of viruses and attacks. For instance, viruses such as \textit{ILOVEYOU} and \textit{My Doom} infected millions of computers through email spreading~\cite{knight2000iloveyou,wong2004studyMydoom}. Such viruses or worms are capable of imposing an enormous loss to a network. For example, a virus which causes servers to malfunction in a cloud network can lead to a costly repairment before the network's revival, so does those that corrupt the transmitted information between power sensors in a smart grid network. To safeguard the network functionality, it is of great importance that infection, or abnormality in the network, can be quickly detected and categorized upon occurrence.
	 Abnormality detection is an important topic in network/system management, and has been studied extensively  (e.g.,~\cite{06:Ayhan:TOIE,08:Yin:TOKDE,10:Li:TOWC}). 
     With respect to the causality, abnormalities can be broadly categorized into two types: random failures and epidemics. In contrast to random failures, an epidemic \textit{spreads} in the network via contacts among entities, and is therefore more hazardous. In the macro level analysis of infections, once abnormalities appear in the network, it is crucial to determine promptly whether they are contagious, since effective measures to combat epidemics, \textit{e.g.}, quarantine, may take some time to deploy while it is a waste of resource if the abnormalities turn out to be random failures. On the other hand, from the perspective of micro level analysis, it is desirable to provide useful information for countermeasure deployment, \textit{e.g.}, which nodes to quarantine first to restrain an epidemic or which network portion to receive priority in vaccination. Therefore, we are motivated to ask the following research questions: (1) \textit{How can we quickly distinguish an epidemic from random failures, based on the observation of the network condition?} (2) In the case of an epidemic, \textit{what information can the observation provide to facilitate deployment of an effective countermeasure?} 
	
	In this regard, there has been some research on the detection of causality for epidemics, e.g., rumor source rooting in online social networks (OSNs)~\cite{Centola:2010:Science,ref:source,ref:source2}, or spreading control in a population~\cite{Lloyd:2001:Science,ref:restrain2,ref:restrain}, where the primary assumption is the existence of an epidemic or a rumor in the network. In these cases, the state evolution of the network is driven by the epidemic spreading process. Stemming from epidemiology, an epidemic process is useful for capturing the spreading behavior when individuals change their states upon having a \textit{contact} with others. However, in addition to rumors or epidemics, individuals in a network can receive ideas or information independently or get sick randomly, which may also lead to irregular behaviors of the nodes. There is little research on differentiating an epidemic from random failures. To this end, Milling \textit{et al.} designed a Median Ball algorithm for the same purpose~\cite{Milling:13:MobiHoc}. This algorithm makes the decision by solely comparing the radius of a ball which contains a certain portion of the infected nodes to a threshold. However, an inherent assumption in~\cite{Milling:13:MobiHoc} is the homogeneity of the network edges\footnote{\label{foot:1}i.e., the epidemic has the same contact probability along all edges.}, which leads to an isotropically concentrated area of infected nodes within a region. This limitation has prompted a follow-up work \cite{15:Milling:InfoCom}, in which two algorithms for the infection detection problem are proposed: \textit{Ball Density} and \textit{Relative Ball Density}. \textit{Ball Density} algorithm makes the decision by comparing the number of infected nodes inside a ball with a certain radius to a threshold. The improved version of it, \textit{Relative Ball Density} (RBD), compares the ratio of the number of infected nodes inside a ball and that outside the ball with a threshold. In these two algorithms, the radius of the ball has to be chosen empirically. Also, due to their inherent mechanism of detection, these algorithms naturally fail in the presence of multiple epidemics in the network or when an epidemic stems from multiple initial seeds. Moreover, although these  two algorithms can handle the heterogeneity of the edges, it is found through simulation (see Section~\ref{sec:numerical}) that they exhibit a poor performance when the underlying graph structure is dense and irregular. These limitations render the ball density-based algorithms ineffective in real-world settings such as the Internet and social networks.
    
	Seemingly simple, differentiating epidemics from random failures is actually a challenging problem due to the following considerations. (1) Networks are generally large and irregular; hence it is difficult or even impossible to track numerous contacts in the network to determine whether the abnormality is spontaneous or being spread from others. (2) From the perspective of a network administrator, available information about the abnormality can be as limited as a single \textit{snapshot} representing the nodes' status, an example of which is illustrated in Figure~\ref{fig:comparisionRFanEP}. This figure depicts two instances of an infected scale-free (SF)\footnote{A popular model for many real world networks including social networks and the Internet~\cite{ref:Internet}.} network~\cite{Albert:02:RevPhys} with $100$ nodes when the infection is due to an epidemic stemming from one source node (left snapshot) and random failures (right snapshot), respectively. Due to the irregular structure of the network, even in this small-scale network, it is challenging to identify the root causes of abnormalities by simple inspection. Also, the existence of false positives (healthy nodes reported as infected) and false negatives (infected nodes reported as healthy), or epidemics originating from multiple source nodes, may further blur the difference between random failures and an epidemic, which adds another level of difficulty to the problem. (3) Straightforward detection measures, e.g., examining the footprints of virus codes, are not applicable when it comes to a different type of virus.

To address these challenges, we are motivated to design a \textit{generic} framework that can accurately \textit{detect} the causality of an abnormality in an irregular heterogeneous network, and provide necessary \textit{information} for countermeasure design, from a single snapshot of the network status. In this regard, we draw a connection between infection analysis and graph signal processing (GSP), mainly developed in~\cite{Shuman:13:ISPM, Sandryhaila:14:TOSP, Sandryhaila:13:TOSP,ref:GW1,ref:GW2} recently, by incorporating the nodes' status into the network structure and interpreting a network snapshot as a \textit{graph signal}. The main goal of this paper is to provide simple, yet effective, online algorithms to conduct infection analysis based on GSP. In summary, to conduct the macro (network-level) analysis: 1) we propose multiple detection metrics based on the graph Fourier Transform (GFT) spectrum, along with their robustness analysis. 2) We develop a new class of graph wavelets called distance-based graph wavelets (DBGWs). 3) We utilize GFT spectrum, spectral graph wavelets~\cite{ref:GW2}, and DBGWs in a machine learning-based framework to carry out the analysis. For the micro (node-level) analysis, countermeasure design based on DBGWs are further investigated. Extensive simulations are conducted to verify the effectiveness of all the proposed algorithms in various network settings.
    
	\begin{figure}[t]
	\centering
			\includegraphics[width=3.5 in,height=1.5 in]{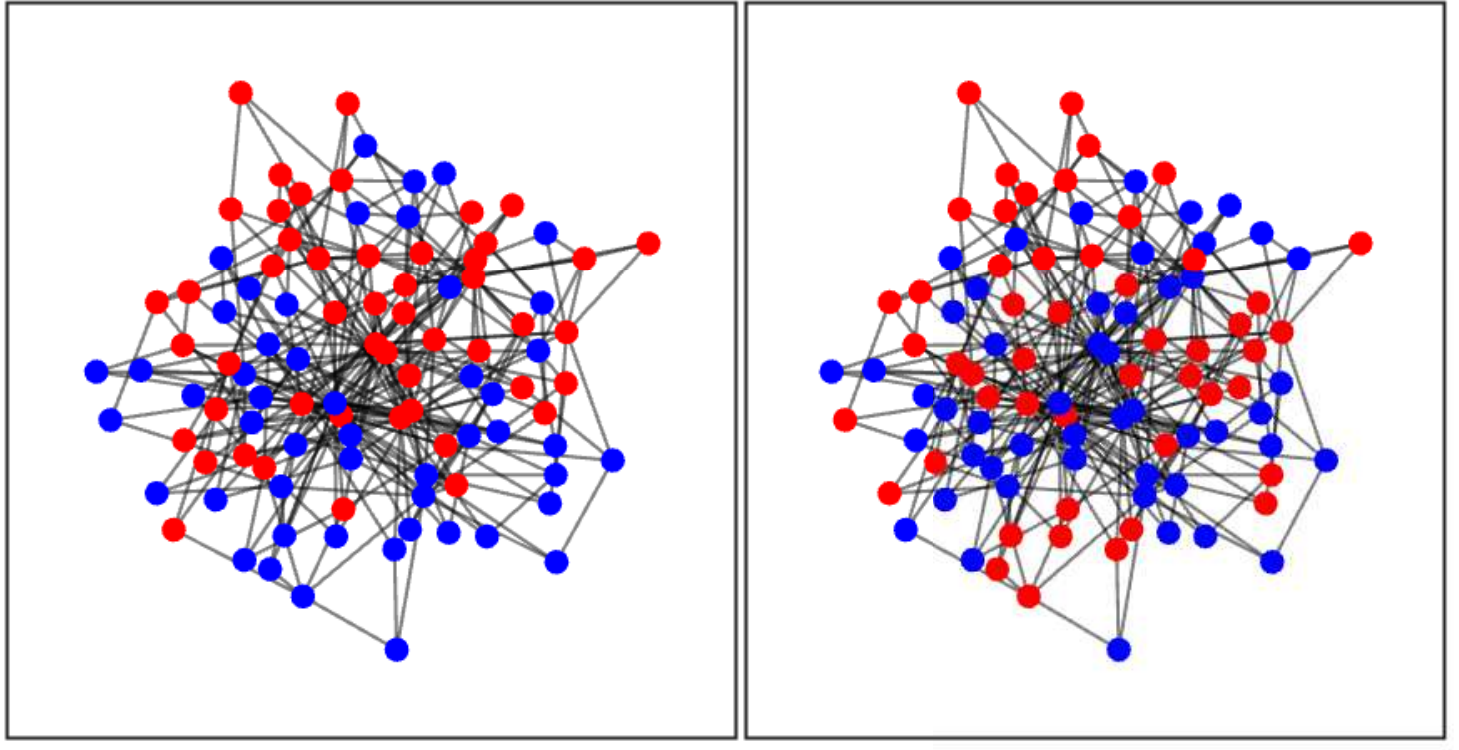}
            
		\caption{Comparison between an epidemic (left) and random failures (right) in an SF network with size $100$ for the same number of infected nodes ($50$). Red nodes are infected while the blue nodes are in a healthy state.}
        \label{fig:comparisionRFanEP}
	\end{figure}
\textbf{Structure of the paper:} System model is explained in Section~\ref{sec:systemModel}, which contains discussions of the network model, spreading process, macro and the micro analysis of infections, and construction of a graph signal based on a snapshot of the network. Basics of the graph Fourier Transform is given in Section~\ref{sec:GSPba}. In Section~\ref{sec:GraphWavelet} a new family of graph wavelets, called distance-based graph wavelets, are proposed and developed for the purpose of infection analysis. Also, spectral graph wavelets are presented in this section. Metric-based and machine learning-based macro analysis of infections are presented in Section~\ref{sec:MacroMetric} and~\ref{sec:MacroLeaning}, respectively. Section~\ref{sec:Micro} is devoted to the micro analysis of infections using DBGWs. Simulation results are presented in Section~\ref{sec:numerical}. Finally, Section~\ref{sec:conclusion} concludes the paper and provides directions for future work. 

	\section{System model}\label{sec:systemModel}
\noindent	In this section, we introduce the terminologies and the models for both the macro and the micro analysis of infections.
    \vspace{-4mm}
	\subsection{Network Model}
	Entities in the system of interest are abstracted to a set of nodes $\mathcal{V}$. Contacts between two entities are described by an edge connecting a pair of nodes, while the frequency of contacts, or equivalently, the probability that the \textit{virus} spreads along an edge is proportional to the \textit{weight} of that edge. A network is represented by a weighted undirected graph $G=(\mathcal{V},\mathcal{E},w,\beta)$, where $\mathcal{V}$ denotes the set of nodes, $\mathcal{E}$ denotes the set of edges, $w$ is the weight function, $w: \mathcal{V}^2 \rightarrow (0,\beta]$, defined for adjacent nodes $i$ and $j$ as: $w(i,j)=\frac{\beta}{a_{ij}}$, $\forall a_{ij}\neq 0$ ($w(i,j)=0$, $\forall a_{ij}= 0$), where $a_{ij}\in \{0\} \cup [1,\infty)$ is an element of the adjacency matrix $\mathbf{A}$ representing the \textit{distance} between the nodes, and $\beta \in (0,1]$ represents the \textit{basic infection rate}. Similarly, $w(i,j)$'s are collected into a weight matrix $\mathbf{W}$. The distance between two nodes can be interpreted differently in various networks, e.g., physical distance between two sensors in a wireless sensor network, strength of the social interaction between two nodes in a social network, frequency of contact between two computers in a computer network, etc. In all the scenarios, it is assumed that a small/large distance corresponds to frequent/rare exchange of information between the individuals. Henceforth, a connected network is assumed.
        \vspace{-4.5mm}
	\subsection{Abnormality Model}
	For an arbitrary index ordering of the nodes in the set $\mathcal{V}=\{v_1, \cdots, v_N\}$, $|\mathcal{V}|=N$, let $x_i(t)$ denote a \textit{label} representing the \textit{state} of node $v_i \in \mathcal{V}$ at time $t$, which can take two possible values: infected ($x_i(t)=``infected"$), or healthy ($x_i(t)=``healthy"$). Upon the occurrence of an epidemic, each node in the set of infected nodes, $v_i \in \mathcal{I}_t \triangleq \{v_i \in \mathcal{V} |x_i(t)=``infected"\}$, attempts to infect any of its susceptible neighbors, $v_j\in \mathcal{N}_{v_i}\setminus \mathcal{I}_t$, $\mathcal{N}_{v_i} \triangleq \{v_j\in \mathcal{V}  | (v_j,v_i) \in \mathcal{E}\}$, with probability of $w(i,j)$ during each time step. In the random failure scenario, irrelevant to the states of other nodes, a healthy node $v \in \mathcal{V}$ may fail (become infected) with a probability $p_{fail}$.
	    \vspace{-4.5mm}
	\subsection{Reporting Model}
	Let $X(t) = \{x_i(t)\}_{v_i \in \mathcal{V}}$ illustrate the \textit{true} states of all the nodes in $\mathcal{V}$. Also, let $r_i(t)$ denote a label representing the \textit{observed/reported} state of node $v_i$ at time $t$. A \textit{snapshot} of the network $\{r_i(t)\}_{v_i \in \mathcal{V}}$ may contain \textit{faulty observations} (false positives and false negatives) due to noise, the observation resolution constraints, or adversarial nodes that deliberately report inaccurate states. We restrict the input of our algorithms to be one single snapshot of the network; and thus omit the time index, and denote a snapshot as $R = \{r_i\}_{v_i \in \mathcal{V}}$.
	    \vspace{-4.5mm}
	\subsection{Macro vs. Micro Analysis of Infections}
	
	In this work, given a snapshot of the network $R = \{r_i\}_{v_i \in \mathcal{V}}$, two categories of analysis are conducted: the macro and the micro analysis. The objective of the macro analysis is to determine \textit{the existence of an epidemic} in a network, which is referred to as \textit{infection detection}. Let the null hypothesis $H_0$ be the absence of an epidemic in the network, and $H_1$ the alternative. The detection performance can be measured by the average infection detection probability ($AIDP$) defined as:
% * <jwang50@ncsu.edu> 2018-02-13T02:35:11.604Z:
%
% > $\gamma$
% Already used in the definition of HECR, how about p?
% ^. Good point. Applied.
	\begin{equation}
	AIDP =  p(H_0|H_0) p_0 +p(H_1|H_1) p_1,
	\end{equation}
	where $p_0$ and $p_1$ are the probabilities of existence of random failures and an epidemic, respectively. In this context, $p({H_1}|H_0)$ denotes the probability of Type I error (a \textit{false alarm}), while $p(H_0|{H_1})$ indicates the probability of Type II error (a \textit{miss}).    
    
After determining the existence of an epidemic, node-level analysis can be conducted to determine an efficient way of applying the countermeasures, which is the goal of the micro analysis of infections. In this work, we focus on two well-known countermeasures: vaccination and quarantine. 
    \vspace{-4mm}
    \subsection{Construction of a Graph Signal based on a Snapshot}\label{subsec:noise}
	
    For an arbitrary labeling order of the nodes in the set $\mathcal{V}=\{v_1, \cdots, v_N\}$, and an arbitrary set of real numbers $\mathcal{S}=\{s_1, \cdots, s_N\}$, where $|\mathcal{S}|=|\mathcal{V}|=N$, a graph signal $\mathbf{S}$ is a one to one mapping between the elements of $\mathcal{V}$ and $\mathcal{S}$ defined as $\mathbf{S}$: $v_i\mapsto s_i , \; 1\leq i \leq N$. In this study, based on the reported snapshot of the network $R = \{r_i\}_{v_i \in \mathcal{V}}$, we pursue a \textit{binary graph signal assignment} and construct the corresponding graph signal as:
	\begin{equation}\label{eq:binarySig}
	\mathbf{S}(v_i) \triangleq \mathbf{S}(i)= \begin{cases}
	A, \quad &\text{if } r_i =``infected",\\
	-A, \quad &\text{if } r_i=``healthy",
	\end{cases}
	\end{equation}
	where $A\in \mathbb{R}^+$.\footnote{Allowing the graph signal $\mathbf{S}$ associated with the snapshot $\{r_i(t)\}_{v_i \in \mathcal{V}}$ to record the time since nodes have become infected enables analysis in a finer grain. However, as shown in simulations, the proposed algorithms can exhibit an excellent accuracy by taking the proposed ``binary" input.} Based on a noisy snapshot, the observed/noisy graph signal $\mathbf{S}^{(n)}$ is the sum of two components: the true (initial) signal value $\mathbf{S}^{(i)}$ and a noise vector $\mathbf{n}= \mathbf{S}^{(n)}-\mathbf{S}^{(i)}$ due to faulty reports.    

	%%%%%%%%%%%%%%%%%%%%%%%%%%%%%%%%%%%%%%%%%%%%%%%%%%%
    %%%%%%%%%%%%%%%%%%%%%%%%%%%%%%%%%%%%%%%%%%%%%%%%%%%%%
    %%%%%%%%%%%%%%%%%%%%%%%%%%%%%%%%%%%%%%%%%%%%%%%%%%%%%%%
    %%%%%%%%%%%%%%%%%%%%%%%%%%%%%%%%%%%%%%%%%%%%%%%%%%%%%%
    \vspace{-2mm}
	\section{Graph Fourier Transform}\label{sec:GSPba}
	\noindent The graph Laplacian matrix $\mathbf{L}$ of a (connected) network represented by $G=(\mathcal{V},\mathcal{E},w,\beta)$ is defined as:
	\begin{equation}
	\mathbf{L} \triangleq \mathbf{D}-\mathbf{W},
	\end{equation}
	where $ \mathbf{D}$ is a diagonal matrix with the $i^{th}$ diagonal element $d_{ii}$ defined as: $d_{ii}=\sum_{j=1}^{N} w(i,j)$. For the sorted eigenvalues $0=\lambda_0 < \lambda_1 \leq \lambda_2\leq  \cdots \leq \lambda_{N-1}= \lambda_{max}$ of matrix $\mathbf{L}$, let $\{\mathbf{u}_0,\cdots,\mathbf{u}_{N-1}\}$ denote the corresponding set of orthonormal eigenvectors. Eigenvectors of the graph Laplacian matrix can be considered as an analogy to the basis functions of the classic Fourier transform $\{e^{j2\pi f x}\}$, where each of them describes a specific frequency. In this paradigm, the associated frequency to an eigenvector is determined by the corresponding eigenvalue. For a graph signal, the notion of frequency can be interpreted as the variation of the signal across the network. As is known~\cite{Shuman:13:ISPM}, as the eigenvalue increases, the corresponding eigenvector exhibits more variations across the nodes in the network. The GFT spectrum $\hat{\mathbf{S}}$ of a graph signal $\mathbf{S} \in \mathbb{R}^N$ is defined as the projection of $\mathbf{S}$ onto the space of graph Laplacian eigenvectors:
    
        \begin{equation}\label{eq:GFTinnerB}
    \hat{\mathbf{S}}=[\hat{{S}}(\lambda_0),\hat{{S}}(\lambda_1),\cdots,\hat{{S}}(\lambda_{N-1})],
    \end{equation}
    where each element is the GFT of signal $\mathbf{S}$ at a certain frequency/eigenvalue, given by:
	\begin{equation}\label{eq:GFTinner}
	\hat{{S}}(\lambda_l)=<\mathbf{S},\mathbf{u}_l>= \sum_{i=1}^{N} \mathbf{S}(i) \mathbf{u}_l(i),\;\; 0\leq l \leq N-1.
	\end{equation}
	Consequently, each element of the GFT spectrum can be interpreted as the similarity of the signal to the corresponding eigenvector. 
	
	%%%%%%%%%%%%%%%%%%%%%%%%%%%%%%%%%%%%%%%%%%%%%%%%%%%
    %%%%%%%%%%%%%%%%%%%%%%%%%%%%%%%%%%%%%%%%%%%%%%%%%%%%%
    %%%%%%%%%%%%%%%%%%%%%%%%%%%%%%%%%%%%%%%%%%%%%%%%%%%%%%%
    %%%%%%%%%%%%%%%%%%%%%%%%%%%%%%%%%%%%%%%%%%%%%%%%%%%%%%
	
	\vspace{-2.99mm}
	\section{Graph Wavelets}\label{sec:GraphWavelet}
	\subsection{Background}
	\noindent The classic continuous wavelet transform (CWT) and its applications have been extensively studied in literature, mainly in the image processing area~\cite{ref:wave1,ref:wave2}. In CWT, wavelets at different locations and scales are constructed based on translating and scaling of a \textit{mother wavelet} $\psi$:
	\begin{equation}
	\psi _{s,a}(x)=\frac{1}{s} \psi\left(\frac{x-a}{s}\right),
	\end{equation}
	where $s$ and $a$ correspond to the scale and location, respectively. Wavelet coefficients at scale $s$ and location $a$ for a function (signal) $f(x)$, $W_f(s,a)$, are defined as:
	\begin{equation}
	W_f(s,a)=\int_{- \infty}^{\infty} \frac{1}{s} \psi^*\left(\frac{x-a}{s}\right)f(x)dx.
	\end{equation}
     The theory of CWT is mainly developed to analyze the signals described in regular and Euclidean spaces.  In the following, we propose a new class of graph wavelets and develop them for the purpose of infection analysis on a network structure.
		%%%%%%%%%%%%%%%%%%%%%%%%%%%%%%%%%%%%%%%%%%%%%%%%%%%
    %%%%%%%%%%%%%%%%%%%%%%%%%%%%%%%%%%%%%%%%%%%%%%%%%%%%%
    %%%%%%%%%%%%%%%%%%%%%%%%%%%%%%%%%%%%%%%%%%%%%%%%%%%%%%%
    %%%%%%%%%%%%%%%%%%%%%%%%%%%%%%%%%%%%%%%%%%%%%%%%%%%%%%
    \vspace{-2mm}
	\subsection{Distance-based Graph Wavelets (DBGWs)}\label{subsec:DBGWs}
	In this work, we develop DBGWs for a circle-shape neighborhood; however, using a similar approach, they can be developed for other cases. The following definitions are presented for the ease of discussion. 
	
	\begin{definition}
		For a given graph $G=(\mathcal{V},\mathcal{E},w,\beta)$, a \textbf{path} with length $n$ between nodes $v_i$ and $v_j$ is defined as a sequence of distinct vertices $P(v_i,v_j)=[v_i, v_{k_1}, v_{k_2}, \cdots, v_{k_{n-1}}, v_j]$, where $w(k_j,k_{j+1}) >0,  j=1, \cdots, n-2$, and $w(i,k_{1})$,$w(k_{n-1},j)>0$. Let $\mathcal{P}_{v_i}^{v_j}$ denote the set of all the paths between nodes $v_i$ and $v_j$.
	\end{definition}
	\begin{definition}
		For a given graph $G=(\mathcal{V},\mathcal{E},w,\beta)$ and an operator $O\in \{+, \times\}$, the \textbf{path weight} (PW) of a given path $P(v_i,v_j)= [v_i, v_{k_1}, v_{k_2}, \cdots, v_{k_{n-1}}, v_j]$ is defined as follows:
		\begin{align}
	 PW(v_i,v_j)|_{P(v_i,v_j),O} = &w(i,{k_{1}})\; O\; w({k_{1}},{k_{2}})\; \nonumber\\
      &O\; w({k_{2}},{k_{3}})  \cdots O\; w({k_{n-1}},j).
		\end{align}
	 Considering this definition, the \textbf{dominant path} (DP) and the \textbf{weight of the dominant path} (WDP) between nodes $v_i$ and $v_j$ are defined as:
		\begin{equation}
		DP(v_i,v_j)|_O = \argmax_{P(v_i,v_j) \in \mathcal{P}_{v_i}^{v_j}} {PW(v_i,v_j)|_{P(v_i,v_j),O}},
		\end{equation}
		\begin{equation}
		WDP(v_i,v_j)|_O = \max_{P(v_i,v_j) \in \mathcal{P}_{v_i}^{v_j} }{PW(v_i,v_j)|_{P(v_i,v_j),O}}.
		\end{equation}
        Also, the \textbf{length of the dominant path} (LDP) is given by:
		\begin{equation}
		LDP(v_i,v_j)|_O = \bigg|DP(v_i,v_j)|_O\bigg|-1,
		\end{equation}
        where $|.|$ denotes the cardinality of a set. 
        Using this terminology, two nodes $v_i$ and $v_j$ are considered to be $r$ \textbf{pseudo-hops} away if $LDP(v_i,v_j)|_O=r$. 
	\end{definition}
	
% 	\begin{definition}
% 		Given a mathematical operator $O$ and a graph $G=(\mathcal{V},\mathcal{E},w,\beta)$, two nodes $v_i$ and $v_j$ are considered to be $r$ pseudo-hops away if and only if $LDP(v_i,v_j)|_O=r$.
% 	\end{definition} 
    \begin{remark}
    As compared to the most of networking-related literature, in this work, the definition of PW, DP, WDP, and LDP are presented based on a mathematical operator $O$. Choosing $O=+$ coincides with the conventional definition of these parameters. However, it is shown below that choosing $O=\times$ captures the spreading nature of an epidemic.
    \end{remark}
    \begin{remark}
    The term ``pseudo-hop" is used to avoid confusion with the concept of ``hop" in graph theory. Note that the pseudo-hop is defined with respect to the dominant path between two nodes. For example, assume that there are only two disjoint paths between nodes $v_i$ and $v_j$. The first path is three hops with weights $1$,$0.9$, and $1$, while the other path is one hop with weight $0.01$. In this case, when $O=\times$, nodes $v_i$ and $v_j$ are $3$-pseudo hopes away. 
    \end{remark}
	
 Let $\mathcal{C}_{v_c}^r$ denote the set of the nodes located at most $r$ pseudo-hops away from $v_c$, which will be referred to as the set of nodes  inside a \textit{disk}  with center $v_c$ and
radius $r$ for simplicity. Note that $\mathcal{C}_{v_c}^0=\{v_c\}$.
    	
	\begin{definition}
		Centering at a node $v_c$, for a specific radius $r\geq 1$, a \textbf{ring} $\mathcal{C}_{v_c}^{\left(r \; \Delta \; (r-1) \right)}$ describes the set of the nodes which are exactly $r$ pseudo-hops away from node $v_c$.\footnote{The mathematical operator $O$ is omitted in the notations for simplicity.} Mathematically:
		    \begin{equation}
		\mathcal{C}_{v_c}^{\left(r \; \Delta \; (r-1) \right)}= \mathcal{C}_{v_c}^{r} \setminus \mathcal{C}_{v_c}^{r-1} =  \mathcal{C}_{v_c}^{r} - \;\mathcal{C}_{v_c}^{r-1}.
		\end{equation}
	\end{definition}

	\begin{definition}\label{def:DBGWdefi}
		For a specified center node $v_c$ and a radius $s\in\mathbb{Z}^+,\;s \geq 1$, the wavelet function of DBGW, $\psi_{v_c,s}: \mathcal{V} \rightarrow \mathbb{R}$, is defined as follows:
		\begin{equation} \label{eq:basis}
		\psi_{v_c,s}^{DBGW}(v)=
		\begin{cases}
		\frac{H(WDP(v_c,v)|_O) m^s_{s'}}{\sum_{n\in \mathcal{C}_{v_c}^{\left(s' \; \Delta \; (s'-1) \right)} }H(WDP(v_c,n)|_O)}  \\
		\hspace{5mm} \forall v \in \mathcal{C}_{v_c}^{\left(s' \; \Delta \; (s'-1) \right)}, 1\leq s' \leq s,  \\
		m^s_{0} \hspace{6mm}\textrm{if}\; v =v_c\; (s'=0),  \\
		0\; \hspace{8mm} \forall v \in \mathcal{C}_{v_c}^{\left(s' \; \Delta \; (s'-1) \right)},  s' > s,
		\end{cases}
		\end{equation}
		where $H:\mathbb{R^+} \rightarrow \mathbb{R^+}$ is an arbitrary increasing function, and  $m^s_{s'}$ is a non-zero constant chosen for each $s'\in\{0,1,\cdots,s\}$, such that $\sum_{s'=0}^{s} m^s_{s'} =0$. Parameter $s$ takes the interpretation of ``scale" as in CWT.
	\end{definition}
    
	\begin{theorem}
		Centering at an arbitrary node $v_c$, for an arbitrary scale $s$, wavelet functions of DBGW satisfy the zero mean property of a mother wavelet. Mathematically:
		\begin{equation}
		\sum_{v \in \mathcal{V}} \psi_{v_c,s}^{DBGW}(v)=0.
		\end{equation} 
	\end{theorem}
	\begin{proof}
		The proof can be obtained using algebraic manipulations based on Eq.~\eqref{eq:basis}. 
	\end{proof}

	\begin{remark}
		For an unweighted graph $G=(\mathcal{V},\mathcal{E})$ and a linear increasing function $H(x)=c x$, $c\in \mathbb{R}^+$, when $O=+$ (assuming the same weight for all the edges), the wavelet functions of DBGW is given by:
		\begin{equation} \label{eq:basisSimplified}
		\psi_{v_c,s}^{DBGW}(v)=
		\begin{cases}
		\frac{m^s_{s'}}{| \mathcal{C}_{v_c}^{\left(s' \; \Delta \; (s'-1) \right)}|}  \\
		\hspace{4mm} \forall v \in \mathcal{C}_{v_c}^{\left(s' \; \Delta \; (s'-1) \right)}, 1\leq s' \leq s,  \\
		m^s_{0} \hspace{6mm}\textrm{if}\; v =v_c\; (s'=0),  \\
		0\; \hspace{8mm} \forall v \in \mathcal{C}_{v_c}^{\left(s' \; \Delta \; (s'-1) \right)},  s' > s,
		\end{cases}
		\end{equation}
		which with some modifications, matches the wavelet functions presented in~\cite{ref:GW1}, suitable for traffic analysis.
	\end{remark}
    Centering at a node $v_c$, for a specified scale $s$, the wavelet function of DBGW assigns non-zero values to the nodes which are located inside the set $\mathcal{C}_{v_c}^s$ and zero value to the rest of the nodes. Consequently, the wavelet function of DBGW captures the signal variation in different neighborhoods through suitable choice of the center locations and scales.
	Centering at an arbitrary node $v_c$, for an arbitrary scale $s$, the wavelet coefficients for a signal $\mathbf{S}\in \mathbb{R}^N$ can be obtained as follows: 

	\begin{flalign}\label{DBGWCoeff}
	\Scale[0.918] {W^{DBGW}_\mathbf{S}(v_c,s)=<\left[\psi^{DBGW}_{v_c,s}(v_1),\cdots, \psi^{DBGW}_{v_c,s}(v_N)\right],\mathbf{S}>.}
	\end{flalign}
A node $v\in \mathcal{V}$ can locally compute its wavelet function at scale $s$ when the set $\mathcal{C}^s_{v}$ along with $WDP(v,v')|_O$, $\forall v'\in \mathcal{C}^s_{v}$ are known at the node. This is a nice property in large-scale networks, since nodes can derive the wavelet functions in a distributed fashion, which notably improves the speed of analysis.
    %%%%%%%%%%%%%%%%%%%%%%%%%%%%%%%%%%%%%%%%%%%%%%%%%%%
    %%%%%%%%%%%%%%%%%%%%%%%%%%%%%%%%%%%%%%%%%%%%%%%%%%%%%
    %%%%%%%%%%%%%%%%%%%%%%%%%%%%%%%%%%%%%%%%%%%%%%%%%%%%%%%
    %%%%%%%%%%%%%%%%%%%%%%%%%%%%%%%%%%%%%%%%%%%%%%%%%%%%%%
    \vspace{-3mm}
	\subsection{DBGW Design for Epidemic Detection}
	\subsubsection{Epidemic Spreading Analysis}
	For a given non-infected node ${v_j}$, the probability of getting the infection at the next time instant from one of its infected neighbor nodes $v_i$ is $w(i,j)$. Similarly, if $v_i$ and $v_j$ are two-hops away, the probability of getting the infection from node $v_i$ in two time instants along the path $P({v_i},v_j)=[v_i,v_k,v_j]$ is $w(i,k) w(k,j)$, where $v_k$ is an initially non-infected node (and gets infected from $v_i$ subsequently). With the same reasoning, we can derive the probability of $v_j$ getting the infection from an infected node $v_i$ along the path $P(v_i,v_j)=[v_i,v_1,\cdots,v_{r-1}, v_j]$ in $r$ time instants, that is, $PW(v_i,v_j)|_{P(v_i,v_j),\times}$. The following theorem is the result of this discussion.

\begin{theorem}\label{th:inverse}
		Assuming the length of the observation time window $T$ and a given snapshot at time $t\leq T$, the probability of a healthy node $v_i \notin \mathcal{I}_t$ getting the infection from node $v_j \in \mathcal{I}_t$ is an increasing function of $WDP(v_i, v_j)|_\times$.
	\end{theorem}
  	%%%%%%%%%%%%%%%%%%%%%%%%%%%%%%%%%%%%%%%%%%%%%%%%%%%
    %%%%%%%%%%%%%%%%%%%%%%%%%%%%%%%%%%%%%%%%%%%%%%%%%%%%%
    %%%%%%%%%%%%%%%%%%%%%%%%%%%%%%%%%%%%%%%%%%%%%%%%%%%%%%%
    %%%%%%%%%%%%%%%%%%%%%%%%%%%%%%%%%%%%%%%%%%%%%%%%%%%%%%
	\subsubsection{Designing the Parameters of DBGWs} The mathematical operator $O$ should allow drawing a nice connection to the spreading aspect of an infection. Considering Theorem \ref{th:inverse}, the operator is chosen as: $O=\times$. Also, according to~Eq.~\eqref{eq:basis}, each element of sequence $\{m_{s'}^s\}^s_{s'=0}$ determines the effect of the graph signal values of nodes located in a certain neighborhood on the amplitude of the wavelet coefficient centered at a node. In particular $\forall r\in [1, s]$, if $m^s_r=0$, that particular neighborhood is neglected, and if $m^s_r m^s_0 >0$ ($m^s_r m^s_0 <0$) the more similarity (contrast) between the signal values of the nodes in that neighborhood and the center node leads to a larger wavelet coefficient at the center node. Also, function $H$ controls the wavelet function amplitude in a certain neighborhood. For the application considered in this work, $H(x)=cx^k$, where $c,k\in \mathbb{Z}^+$ is in general a good choice. 	
	\begin{example}
		Given a graph signal, the goal is to measure the contrast between the signal values of a disk with radius $s-1$ and the signal values of an outer ring with radius $s$, $s\geq2$, around each node. In this case, a good candidate sequence for $\{m_{s'}^s\}_{s'=0}^s$ is choosing $m_{0}^s,\cdots,m_{s-1}^s>0$ while $m_{s}^s<0$ such that $\sum_{s'=0}^s m_{s'}^s=0$. Figure \ref{diag:EgWavelet} depicts a graph and the values of the wavelet function of DBGW\footnote{It can be seen that nodes with larger WDP to the center node take larger absolute values.} centered at node $a$ at scale $2$, where $m^2_{0}=0.7$, $m^2_{1}=0.3$, $m^2_{2}=-1$, and $H(x)=x$. According to Eq.~\eqref{DBGWCoeff}, $W^{DBGW}_\mathbf{S}(a,2)$ reaches its maximum (in the absolute value sense) when the signal values of all the nodes inside $\mathcal{C}^1_a$ is the opposite of those located at $\mathcal{C}^{2 \Delta 1}_a$.
	\end{example} 
    \begin{remark}\label{rem4}
    The sequence $\{m_{s'}^s\}^s_{s'=0}$ can be obtained using a zero-mean continuous real function supported in the interval $[0,1]$. In this case, for a given $s'$, $m_{s'}^s$ corresponds to the value of the integral of the function in the interval $[\frac{s'}{s+1},\frac{s'+1}{s+1}]$.
    \end{remark}
\begin{figure}[t]
		\includegraphics[width=0.3\textwidth]{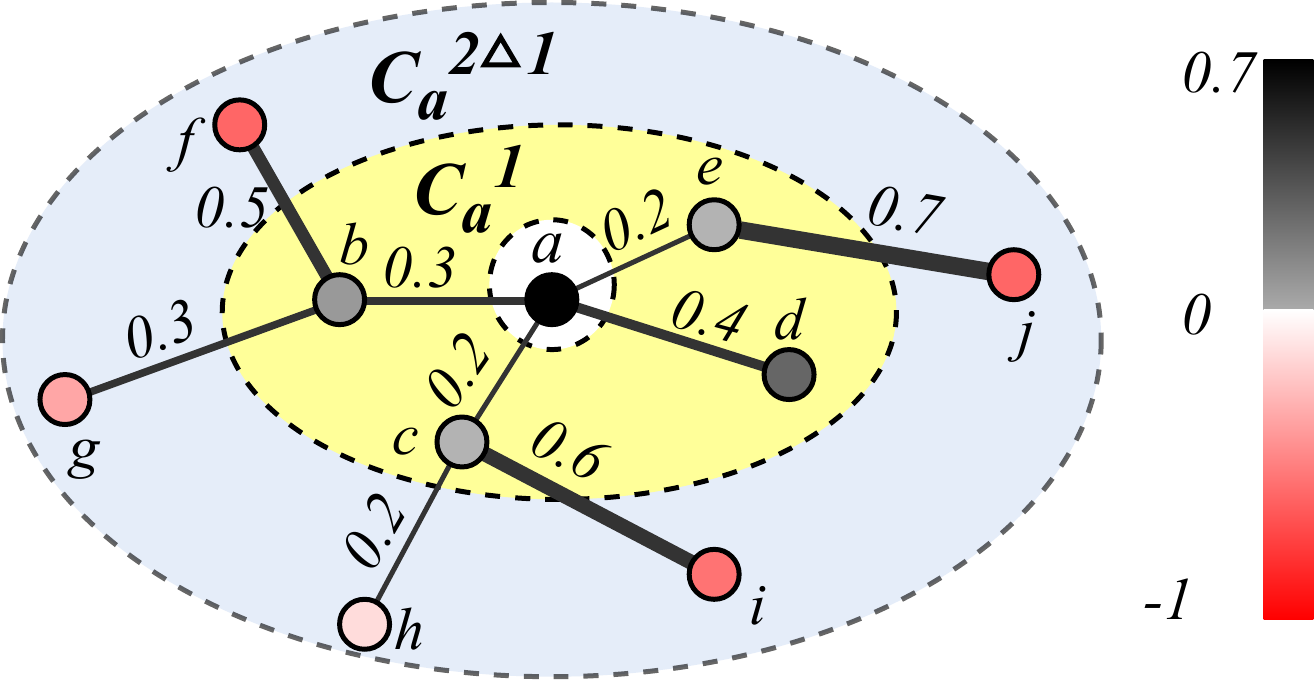}
		\centering
		\caption{An example of wavelet function of DBGW. The weight of each edge is indicated numerically and also in the thickness of the connection. The color of each vertex depicts the value of the wavelet function at the node. $\mathcal{C}_a^1 = \{a,b,c,d,e\}$ and $\mathcal{C}_a^{2\Delta 1} = \{f,g,h,i,j\}$. }
		\label{diag:EgWavelet}
\end{figure}

		%%%%%%%%%%%%%%%%%%%%%%%%%%%%%%%%%%%%%%%%%%%%%%%%%%%
    %%%%%%%%%%%%%%%%%%%%%%%%%%%%%%%%%%%%%%%%%%%%%%%%%%%%%
    %%%%%%%%%%%%%%%%%%%%%%%%%%%%%%%%%%%%%%%%%%%%%%%%%%%%%%%
    %%%%%%%%%%%%%%%%%%%%%%%%%%%%%%%%%%%%%%%%%%%%%%%%%%%%%%
	\vspace{-3mm}
	\subsection{Spectral Graph Wavelets (SGWs)}
    Consider a weighted graph $G=(\mathcal{V},\mathcal{E},w,\beta)$, where $|\mathcal{V}|=N$. Spectral graph wavelets (SGWs)~\cite{ref:GW2} are defined based on the choice of a kernel function $g: \mathbb{R^+} \rightarrow \mathbb{R^+}$, which behaves similarly to a band-pass filter:
	\begin{gather}
	g(0)=0, 
	\lim_{x \rightarrow \infty} g(x)=0.
	\end{gather}
	In this case, the graph wavelet functions at scale $s$ centered at vertex $v_c\in \mathcal{V}$ can be defined as:
	\begin{equation}
	\psi^{SGW}_{v_c,s}(v)= \sum_{l=0}^{N-1} g(s \lambda_l) \mathbf{u}_l(v_c) \mathbf{u}_l(v), \;\; v\in \mathcal{V},
	\end{equation}
	where $\mathbf{u}_l$'s are the eigenvectors of the Laplacian matrix defined in Section~\ref{sec:GSPba}. Centering at a node $a\in \mathcal{V}$, for an arbitrary scale $s$, the wavelet coefficients of a signal $\mathbf{S}\in \mathbb{R}^N$ is given by: 

	\begin{flalign}
	\Scale[0.938]{ W^{SGW}_\mathbf{S}(v_c,s)=<\left[\psi^{SGW}_{v_c,s}(v_1),\cdots, \psi^{SGW}_{v_c,s}(v_N)\right],\mathbf{S}>.}
	\end{flalign}
	
   As compared to the neighborhood-based filtering approach in DBGWs, the nodes' wavelet coefficients of SGWs are constructed based on filtering the eigenvectors via the kernel function. Hence, in general, SGWs do not provide the desired node-level information of an infection. Also, similar to GFT, SGW design requires the knowledge of Laplacian matrix eigenvectors. We will combine this construction of graph wavelets with machine learning techniques and use them for the macro analysis of infections in Section~\ref{sec:MacroLeaning}.
		%%%%%%%%%%%%%%%%%%%%%%%%%%%%%%%%%%%%%%%%%%%%%%%%%%%
    %%%%%%%%%%%%%%%%%%%%%%%%%%%%%%%%%%%%%%%%%%%%%%%%%%%%%
    %%%%%%%%%%%%%%%%%%%%%%%%%%%%%%%%%%%%%%%%%%%%%%%%%%%%%%%
    %%%%%%%%%%%%%%%%%%%%%%%%%%%%%%%%%%%%%%%%%%%%%%%%%%%%%%
	\vspace{-2.5mm}
	\section{Metric-based Macro Analysis of Infections}\label{sec:MacroMetric}
    	%%%%%%%%%%%%%%%%%%%%%%%%%%%%%%%%%%%%%%%%%%%%%%%%%%%
    %%%%%%%%%%%%%%%%%%%%%%%%%%%%%%%%%%%%%%%%%%%%%%%%%%%%%
    %%%%%%%%%%%%%%%%%%%%%%%%%%%%%%%%%%%%%%%%%%%%%%%%%%%%%%%
    %%%%%%%%%%%%%%%%%%%%%%%%%%%%%%%%%%%%%%%%%%%%%%%%%%%%%%
	\subsection{Methodology}
    Given a network and the corresponding graph Laplacian matrix and its eigenvectors, the eigenvector corresponding to $\lambda_0=0$ has the same value among all the nodes, and thus captures the ``DC" value of the signal. In general~\cite{Shuman:13:ISPM}, small/large eigenvalues correspond to slow/fast variation in the magnitude across the components of the corresponding eigenvector (equivalently, across the nodes in the network). Therefore, upon projection of a graph signal onto the eigen-space of the graph Laplacian matrix (see Eq.~\eqref{eq:GFTinner}), the low-end components in the GFT spectrum (i.e., $\hat{S}(\lambda_i)$ with small values of $i$ in Eq.~\eqref{eq:GFTinnerB}) intuitively capture the low-frequency components of the graph signal (across the network), and vice versa. Considering this fact, we propose the following metrics to capture the signal behavior from the spectrum of the GFT.
	
	\begin{definition}
		For a signal $\mathbf{S} \in \mathbb{R}^N$ with the GFT spectrum $\hat{\mathbf{S}}$, the \textbf{energy of the spectrum} $Eng(\hat{\mathbf{S}})$ is defined as the sum of the magnitude of the frequency components, i.e., 
		\begin{equation}
		Eng(\hat{\mathbf{S}}) = ||~ \hat{\mathbf{S}} ~||_1.
		\end{equation}
	\end{definition}
	
%\footnote{This definition is suitable for the application %considered in this work. The classic definition of the %energy of a real-valued signal $v$ is %$\left(||v||_2\right)^2$.}

For a binary signal $\mathbf{S}=[A,-A,A,\cdots]$, $A\in \mathbb{R}\setminus \{0\}$, the normalized signal can be derived as: $\tilde{\mathbf{S}}=\frac{\mathbf{S}}{A}$. Similarly, the normalized GFT spectrum and its energy are given by: $\widetilde{\hat{\mathbf{S}}}=\frac{\hat{\mathbf{S}}}{A}$ and $\widetilde{Eng}(\hat{\mathbf{S}})=\frac{Eng(\hat{\mathbf{S}})}{A}$, respectively.
	
	\begin{definition}
		The \textbf{$\alpha-$ high-frequency energy concentration ratio} $HECR_{\alpha} (\hat{\mathbf{S}})$, where $\alpha \in [\frac{1}{N},1]$, is defined as:
		\begin{equation}\label{eq:HECR2}
		HECR_{\alpha} (\hat{\mathbf{S}}) = \frac{\sum_{i = \floor{N(1-\alpha)}}^{N-1} | \hat{\mathbf{S}}(\lambda_i)|}{Eng(\hat{\mathbf{S}})}.
		\end{equation}
	\end{definition}
	\begin{definition}
		The \textbf{$\gamma-$ low-frequency energy concentration ratio} $LECR_{\gamma} (\hat{\mathbf{S}})$, where $\gamma \in [\frac{1}{N},1]$, is defined as:
		\begin{equation}\label{eq:LECR2}
		LECR_{\gamma} (\hat{\mathbf{S}}) = \frac{\sum_{i = 0}^{\ceil{\gamma N}-1} |\hat{\mathbf{S}}(\lambda_i)|}{Eng(\hat{\mathbf{S}})}.
		\end{equation}
	\end{definition}
	If there is an epidemic in the network, nodes close to the source will possess similar signal values. Thus, the constructed graph signal will have low variations among close nodes, which corresponds to a low frequency graph signal. In contrast, when the network is affected by random failures, the snapshot contains randomly positioned failed nodes and the nodes with similar signal values would tend to be spread throughout the whole network randomly. This is congruous with a high frequency graph signal. Therefore, the energy of the spectrum is focused more on the higher/lower components in the case of random failures/an epidemic, which is captured using the energy concentration metrics. For the macro analysis of infections, the \textit{smoothness} metric, presented in~\cite{Shuman:13:ISPM}, is also considered.
	
	\begin{definition}
		Smoothness of a graph signal $\mathbf{S}$ with respect to the underlying graph $G=(\mathcal{V},\mathcal{E},w,\beta)$ is defined as:
		\begin{equation}\label{eq:smoothness}
		SM (S)= \sum_{v_i \in  \mathcal{V}} \Big[ \sum_{v_j\in \mathcal{N}_{v_i}} a_{i,j}[\mathbf{S}(j)-\mathbf{S}(i)]^2 \Big]^\frac{1}{2}.
		\end{equation}
	\end{definition}
% The smoothness metric is defined based on the one-hop neighbors of a node. 

% However, especially for graphs with large diameter, we may need the information from the nodes within a certain distance of a node to make a better decision on the smoothness of the signal. For this purpose, we define a generalization of this metric in the following.
% 	\begin{definition}
% 		Generalized smoothness of a graph signal $S$ with respect to the underlying graph $G=(\mathcal{V},\mathcal{E},w,\beta)$ is defined as:
% 		\begin{equation}\label{eq:Gsmoothness}
% 		G\_SM (S,\eta)= \sum_{v_i \in  \mathcal{V}} \Big[ \sum_{v_j\in \mathcal{N}^\eta_{v_i}} b_{i,j} [S(j)-S(i)]^2 \Big]^\frac{1}{2},
% 		\end{equation}
% 		where $b_{i,j}$ describes the shortest distance between the nodes $v_i$ and $v_j$ in terms of number of hops, and
% 		\begin{equation}
% 		\mathcal{N}^\eta_{v_i}= \{v_j\in \mathcal{V} |b_{i,j}  \leq \eta\},
% 		\end{equation}
% 		where $\eta$ is an exploration parameter. Specially, if the exploration parameter is set to be the diameter of the graph, a node's signal value will be compared to those of the rest.
% 	\end{definition}
    
	In the case of an epidemic, as compared to random failures, this metric receives a smaller value.
    	%%%%%%%%%%%%%%%%%%%%%%%%%%%%%%%%%%%%%%%%%%%%%%%%%%%
    %%%%%%%%%%%%%%%%%%%%%%%%%%%%%%%%%%%%%%%%%%%%%%%%%%%%%
    %%%%%%%%%%%%%%%%%%%%%%%%%%%%%%%%%%%%%%%%%%%%%%%%%%%%%%%
    %%%%%%%%%%%%%%%%%%%%%%%%%%%%%%%%%%%%%%%%%%%%%%%%%%%%%
\vspace{-3.5mm}
	\subsection{Algorithm}
 We present the following two-phase algorithm to carry out infection detection using the defined metrics:\\
 \textbf{\textit{Offline Phase (Training):}}
 
 \textbf{Step 1:} Given a training dataset consisting of snapshots of random failures, for each snapshot, build the graph signal according to Eq.~\eqref{eq:binarySig} with a unified signal amplitude $A$.
 
 \textbf{Step 2:} For each metric, derive the \textit{$\epsilon\%$-prediction interval} capturing the metric value in $\epsilon \%$ of the testing scenarios. \\
\textbf{\textit{Online Phase (Prediction):}}

   \textbf{Step 1:} Given a snapshot of the network, build the graph signal according to Eq.~\eqref{eq:binarySig} using the signal amplitude $A$. 
   
   \textbf{Step 2:} For a given metric, calculate the metric value of the constructed graph signal.
   
   \textbf{Step 3:} If the calculated value is outside the determined interval, report an epidemic; otherwise report random failures.
  
  In the Appendix~\ref{app:A}, we conduct robustness analysis to reveal the performance of this algorithm upon having a noisy snapshot. 
   
   \vspace{-1mm}
   \section{Machine Learning-based Macro Analysis of Infections}\label{sec:MacroLeaning}
  \subsection{Methodology}     
       In this case, we apply machine learning techniques with two well-known classifiers: naive Bayes (NB)~\cite{ref:naivebayes} and random forest (RF)~\cite{ref:randomForrest}. The following three vectors may be considered as the \textit{primary features} for the learning/prediction:
\begin{align}\label{eq:GFTfeat}   
      & \mathbf{\hat{\mathbf{S}}}=[\hat{{S}}(\lambda_0), \hat{{S}}(\lambda_2), \cdots, \hat{{S}}(\lambda_{N-1})],\\
	&\Lambda^{DBGW}_{\mathbf{S},s} = [W^{DBGW}_\mathbf{S}(v_1,s), \cdots, W^{DBGW}_\mathbf{S}(v_N,s)],\\
	    &\Lambda^{SGW}_{\mathbf{S},s} = [W^{SGW}_\mathbf{S}(v_1,s), \cdots, W^{SGW}_\mathbf{S}(v_N,s)],
    \end{align}
where $\mathbf{\hat{\mathbf{S}}}$ is the GFT spectrum, while $\Lambda^{DBGW}_{S,s}$ and $\Lambda^{SGW}_{S,s}$ correspond to the collection of the nodes' wavelet coefficients of the DBGW and the SGW, respectively, discussed in Section~\ref{sec:GraphWavelet}. With one approach of feature extraction (GFT, DBGW, or SGW), the RF and the NB classifiers are utilized along with two methods of feature feeding: direct manner and fast approach. In the direct manner, the classifier is fed with the primary features given above. In the fast approach, for each primary feature vector, we obtain three secondary features: the variance, the average, and the interquartile range, i.e., the difference between the $75$-th and $25$-th percentiles of the data (which can be calculated in $O(N)$ complexity using efficient \textit{selection algorithms}), and feed the classifier with them.

       \subsection{Algorithm} The following machine learning-based algorithm is proposed to conduct the macro analysis of infection:\\
        \textbf{\textit{Offline Phase (Training):}} 
 
 \textbf{Step 1:} Given a training dataset of snapshots of both random failures and epidemics, for each snapshot, build the graph signal according to Eq.~\eqref{eq:binarySig} with a unified signal amplitude $A$. 
 
 \textbf{Step 2:}  For each scenario, opting one method of feature extraction, extract the primary features.
   
 \textbf{Step 3:} Selecting one method of feature feeding, train the NB and RF classifiers.\\
 \textbf{\textit{Online Phase (Prediction):}}

   \textbf{Step 1:} Given a snapshot of the network, build the graph signal according to Eq.~\eqref{eq:binarySig} using the signal amplitude $A$. 
   
   \textbf{Step 2:} Based on the derived graph signal, obtain the prediction features via the same approach used in the training phase, and feed them to the trained classifiers to obtain the prediction result.
    
%     The GFT spectrum vector, $
%    [\hat{\mathbf{S}}(\lambda_0), \hat{\mathbf{S}}(\lambda_2), \cdots, \hat{\mathbf{S}}(\lambda_{N-1})]$,
%    can be considered as the \textit{features} for describing a given signal. In this manner, an epidemic and random failures can be differentiated by applying machine learning techniques using this vector as the learning/prediction feature. For this purpose, we use two well-known classifiers: naive Bayes (NB)~\cite{ref:naivebayes} and random forest (RF)~\cite{ref:randomForrest}. The training phase can be done offline using a given dataset of random failures and epidemics. Afterwards, the trained classifiers can be used online to differentiate random failures from an epidemic. In this approach, each experiment (a graph signal corresponding to a snapshot) is associated with $N$ features, ($N=|\mathcal{V}|$). This slows the learning and prediction speed in large-scale networks. To overcome this challenge, an effective sampling approach is discussed in Section~\ref{sec:numerical}. 
    
		%%%%%%%%%%%%%%%%%%%%%%%%%%%%%%%%%%%%%%%%%%%%%%%%%%%
    %%%%%%%%%%%%%%%%%%%%%%%%%%%%%%%%%%%%%%%%%%%%%%%%%%%%%
    %%%%%%%%%%%%%%%%%%%%%%%%%%%%%%%%%%%%%%%%%%%%%%%%%%%%%%%
    %%%%%%%%%%%%%%%%%%%%%%%%%%%%%%%%%%%%%%%%%%%%%%%%%%%%%%
	\subsection{Time Complexity Analysis of the Algorithms}
	For a prior unknown network structure, to utilize GFT and SGW, one needs to derive the complete set of eigenvectors of the Laplacian matrix using the distributed algorithm in~\cite{ref:eig1}, or singular value decomposition methods in the worst complexity of $O(N^3)$. Also, for DBGWs, each node requires to perform (in the worst case) $O(N)$ computations (for a linear function $H$) to derive the DBGW wavelet function centered at itself. This can be done in parallel for all nodes when they have the required information described in~\ref{subsec:DBGWs}. Nevertheless, for a fixed and known network structure, the derivation of wavelet functions and the eigenvectors of the Laplacian matrix can be done offline, which is neglected in the following analysis.  
	
	Given a signal $\mathbf{S}\in \mathbb{R}^N$ and a scale $s$, the worst complexity of deriving $\hat{\mathbf{S}}$, $\Lambda^{DBGW}_{\mathbf{S},s}$ and $\Lambda^{SGW}_{\mathbf{S},s}$ is $O(N^2)$ for each. Note that using parallelization, $\Lambda^{DBGW}_{\mathbf{S},s}$ can be derived in $O(N)$ worst complexity at each node. Hence, the complexity of the learning-based algorithm is mostly determined by the machine learning part. For a dataset of $k$ experiments, each with $f$ features, the learning complexity of the NB and the RF classifiers are $O(kf)$ and $O\left(n_{tree}k^2f'\log(k)\right)$, respectively, where $n_{tree}$ is the number of trees, and $f'$ indicates the number of samples at each node for the RF model, usually chosen as: $f'=O(\sqrt{f})$ \cite{ref:RFcomp}.
	The direct method of feature feeding requires no computation on the entire vector of features; however, it results in a large number of features $f$ upon handling large-scale networks. The fast method needs the calculation of three simple statistics of the data, all of which can be done efficiently in $O(N)$, and it leads to three features for training/prediction. Given the complexity of the machine learning methods mentioned above, the fast method leads to a faster training/prediction. For instance, for a network of size $5000$, the fast method leads to $99.94\%$ reduction in the number of features. Also, the proposed fast method is simpler/faster to implement as compared to the classic dimensionality reduction methods, such as principle component analysis, which require eigen-space decomposition of the collected instances of the training features.  The complexity of deriving the energy concentration-based metrics and smoothness is $O(N^2)$. Also, the detection based on these metrics is straightforward to implement.	
    
    The required basis functions of the GFT, i.e., Laplacian eigenvectors, can be obtained in a distributed fashion using the iterative algorithm in \cite{ref:eig1}. At each iteration of this algorithm, to obtain $k$ \textit{principal eigenvectors}, passing the messages of size $O(k^3)$ between the neighbor nodes along with performing $O(k^3)$ computations at each node is required. In this regard, the proposed machine learning-based algorithm is capable of using a portion of the GFT spectrum for conducting the macro analysis. This aspect is further illustrated in Section~\ref{sec:numerical}.
	
		%%%%%%%%%%%%%%%%%%%%%%%%%%%%%%%%%%%%%%%%%%%%%%%%%%%
    %%%%%%%%%%%%%%%%%%%%%%%%%%%%%%%%%%%%%%%%%%%%%%%%%%%%%
    %%%%%%%%%%%%%%%%%%%%%%%%%%%%%%%%%%%%%%%%%%%%%%%%%%%%%%%
    %%%%%%%%%%%%%%%%%%%%%%%%%%%%%%%%%%%%%%%%%%%%%%%%%%%%%%

	\section{Micro Analysis of Infections}\label{sec:Micro}
	\noindent For the micro analysis of infections, we consider two epidemic countermeasures: vaccination and quarantine. Vaccination refers to the process in which a portion of non-infected nodes become immune to the infection, e.g., by installing an anti-virus software capable of detecting and eliminating the ongoing virus. On the other hand, quarantine refers to the process in which a portion of infected nodes lose their spreading potential, e.g., as some infected computers are identified and restrained from information exchange. The goal of the micro analysis is to find the most impactful infected nodes to quarantine, while determining the most susceptible healthy nodes to vaccinate. There are many ways to define the impact and susceptibility of nodes in a network. In this study, the most impactful/susceptible nodes refer to the infected/healthy nodes completely surrounded by the nodes in the opposite state, which can be identified using the DBGWs wavelet coefficients. The following theorem and corollary provide the foundation for this application of DBGWs.
	
	\begin{theorem}\label{th:sign}
		Given a snapshot of a network, assume the binary graph signal assignment. For an arbitrary scale $s$, if $\sum_{r'=1}^{s}|m^s_{r'}|\leq |m^s_0|$, the DBGW wavelet coefficients of all the nodes with the same state have the same sign.
	\end{theorem}
 \begin{proof}      Let $\mathcal{I}$ denote the set of infected nodes. Since $ m_{0}^s = -\sum_{r'=1}^s m_{r'}^s$, $W^{DBGW}_\mathbf{S}(a,s)$ for any infected node $a\in \mathcal{V}$ is given by:
    \small
\begin{equation}
		\begin{aligned}
		&W^{DBGW}_\mathbf{S}(a,s)=<\psi^{DBGW}_{a,s},\mathbf{S}>=Am_{0}^s +\\
        &\frac{A m_{1}^s}{\sum_{a'\in \mathcal{C}_{a}^{1 \Delta 0}} H\left(WDP(a,a')\right)} \Bigg[\sum_{b \in \mathcal{C}_{a}^{1 \Delta 0}, b\in \mathcal{I}} H\left(WDP(a,b)\right) \nonumber \\
		& -
		\sum_{b \in \mathcal{C}_{a}^{1 \Delta 0} , b \notin \mathcal{I}} H\left(WDP(a,b)\right)   \Bigg]+ \\
        &\vdots \\
        &+\frac{Am_{s}^s}{{\displaystyle \sum_{a'\in \mathcal{C}_{a}^{s\Delta (s-1)}}} H\left(WDP(a,a')\right)} \Bigg[  {\displaystyle \sum_{b \in \mathcal{C}_{a}^{s\Delta (s-1)} , b\in \mathcal{I}}} H\left(WDP(a,b)\right)\\
        &-
		\sum_{b \in \mathcal{C}_{a}^{s\Delta (s-1)} , b \notin \mathcal{I}} H\left(WDP(a,b)\right)   \Bigg].
		\end{aligned}
        \end{equation}
		Since $\sum_{r'=1}^{s}|m^s_{r'}|\leq |m^s_0|$, we can get: 
		\begin{equation}\label{eq:inproofthesign}
        \begin{aligned}
		& 0\leq W^{DBGW}_\mathbf{S}(a,s)\leq 2A m^s_{0}\;\;\textrm{if}\; m^s_{0}>0, \\ 
        & 2A m^s_{0}\leq W^{DBGW}_\mathbf{S}(a,s)\leq 0\;\;  \textrm{if}\; m^s_{0}<0.
		\end{aligned}	
        \end{equation}
        \normalsize
        The same result can be proven for healthy nodes.
	\end{proof}
\normalsize
\begin{remark}
        Note that the condition $\sum_{r'=1}^{s}|m^s_{r'}|\leq |m^s_0|$ can be satisfied by a proper choice for this sequence during the design phase (also see Remark~\ref{rem4}).
\end{remark}
		\begin{corollary}\label{cor:amp}
		(Countermeasure design methodology) Given the declared condition in Theorem~\ref{th:sign}, when $s=1$ and $m^1_0 >0$, the most impactful/susceptible nodes are the infected/healthy nodes with the largest/smallest DBGW coefficients.
	\end{corollary}
	\begin{proof}
		This can be validated by deriving the conditions which lead to the extreme bounds given in Eq.~\eqref{eq:inproofthesign}. In this case, $m_1^0=-m_1^1$; hence for any infected node $a\in \mathcal{V}$,  $ W^{DBGW}_\mathbf{S}(a,1)=
		\Scale[0.9]{Am_{0}^1 +\\
        \frac{A m_{1}^1}{\displaystyle \sum_{a'\in \mathcal{C}_{a}^{1 \Delta 0}} H\left(WDP(a,a')\right)} \Bigg[\displaystyle \sum_{b \in \mathcal{C}_{a}^{1 \Delta 0}, b\in \mathcal{I}} H\left(WDP(a,b)\right) \nonumber }\\
	\Scale[0.9]{	 -
		\displaystyle \sum_{b \in \mathcal{C}_{a}^{1 \Delta 0} , b \notin \mathcal{I}} H\left(WDP(a,b)\right)   \Bigg]}$ reaches its maximum ($2Am_0^1$) when $ b \notin \mathcal{I} \; \forall b \in \mathcal{C}_{a}^{1 \Delta 0}$. The proof for the healthy nodes is similar.
	\end{proof}
	\begin{remark}\label{rem6}
	Note that with the given conditions in Corollary~\ref{cor:amp}, as the portion of healthy/infected nodes with close proximity to an infected/healthy node increases, the value of the wavelet function increases/decreases. In this regard, we can generalize the concept to define highly impactful/susceptible nodes, which correspond to large/small DBGW wavelet coefficients. In the case of $m_0^{1}<0$, the signs of the wavelet coefficients for infected nodes and healthy nodes will be reversed. More discussions concerning micro analysis of infections are given in~\ref{sub:mi}. 
	\end{remark}
%\begin{remark}
        %The above-mentioned approach is based on considering the %density of the nodes in an opposite state of a node to %measure its susceptibility/impactfulness. To account for %the number of neighbours with high WDPs in opposite state %to a node, one can consider the multiplication of the %DBGW wavelet function and the number of nodes located in %the considered neighborhood. 
%\end{remark}
In addition to the countermeasure design, Theorem~\ref{th:sign} provides insights for another potential application of DBGWs, which is determining the origin of an epidemic. This is mainly related to finding the source of a rumor in OSNs, which is out of the scope of this work and presented here as a conjecture. 
% * <jwang50@ncsu.edu> 2018-02-13T02:56:28.222Z:
%
% > homogeneous
% I don't understand what homogeneity has to do with the application. I thought the proposed method should work in both homogeneous and heterogeneous networks, right?
% ^. We can discuss.\\
\\
	{\textbf{Conjecture:}} In a large-scale, (semi-)homogeneous network structure, the initial source of an epidemic is always among the nodes with the lowest (in the absolute value sense) DBGW coefficients at any given scale. Hence, the source can be identified by gathering the DBGWs coefficients at multiple scales.
    
		%%%%%%%%%%%%%%%%%%%%%%%%%%%%%%%%%%%%%%%%%%%%%%%%%%%
    %%%%%%%%%%%%%%%%%%%%%%%%%%%%%%%%%%%%%%%%%%%%%%%%%%%%%
    %%%%%%%%%%%%%%%%%%%%%%%%%%%%%%%%%%%%%%%%%%%%%%%%%%%%%%%
    %%%%%%%%%%%%%%%%%%%%%%%%%%%%%%%%%%%%%%%%%%%%%%%%%%%%%%

	\section{Numerical Results}\label{sec:numerical}
	\subsection{Macro Analysis of Infections}\label{subsec:Macronumerical}
	\noindent In the simulation, two of the most popular random graph structures, Erd{\"o}s–Rényi (ER) \cite{ref:ERGraph}, and Scale-free (SF) graphs \cite{ref:SFGraph,Albert:02:RevPhys}, along with a real-world email network (email-Eu-core) \cite{snapnets,ref:emailnet}, are considered, all of which are described in Table~\ref{tabel:datasets}. ER graphs are associated with a parameter $\rho$, which represents the probability of an edge existing between any two nodes of the graph. The SF graph is constructed based on the so-called \textit{preferential attachment }mechanism, where the construction process of the graph starts with a small connected component and at each step one node is added to the network and connected to $\Gamma$ existing nodes (with probability proportionally to node degrees). For the email-Eu-core network, we use the undirected version of the network presented in \cite{snapnets} while eliminating the island components. For each graph, the distances of edges are chosen independently from a uniform distribution in the interval $[1,100]$, and the value of $A$ in Eq.~\eqref{eq:binarySig} is set to $+1000$. The AIDP (defined in Section~\ref{sec:systemModel}) results are presented with respect to the number of infected nodes in the network, denoted by $\phi$ in the following figures.
	
	For performance evaluation, a dataset consisting of $4000$ simulations, half for epidemics and half for random failures, is constructed. The basic infection rate $\beta$ is set to $0.5$, and $10\%$ of false positives and $10\%$ of false negatives are added to the snapshot. For metric-based detection, an external dataset consisting of $2000$ random failures is used to derive the $\epsilon\%$-prediction interval, where $\epsilon= 95$. In simulations, the results of the $LECR$ and $HECR$ are combined and presented as one metric called \textit{energy concentration}, where the algorithm reports an epidemic if both the $LECR$ and the $HECR$ detect an epidemic, and random failures otherwise. 
%     For the GSL algorithm, the first element of the GFT spectrum $\hat{\mathbf{S}}(\lambda_0)$ is excluded since it only captures the DC value of the signal. 
    For the learning-based algorithms, $10$-fold cross-validation is used, in which the dataset of $2000$ epidemics and $2000$ random failures is evenly divided into $10$ sub-datasets, each with $200$ epidemics and $200$ random failures. In this manner, $10$ experiments are done, where at each experiment $9$ sub-datasets are used for the purpose of training and one is used for prediction. The reported $AIDP$ is the average accuracy of prediction of $10$ experiments. For DBGWs, $s=1$, $H(x)=x^2$, and the sequence $\{m_{r'}^1\}_{r'=0}^1$ is derived based on a zero mean one-sided function supported at the interval $[0,1]$ described as: $\Scale[0.85]{\frac{-4}{\pi^{1/4} \sqrt{3 \sigma}}    \left(\frac{16t^2}{\sigma^2}-1\right)\exp\left(-\frac{8t^2}{\sigma^2}\right)}$ with $\sigma=0.9$ ($m_{0}^1=0.155$, $m_{1}^1=-0.155$). For SGW, to have a proper band-pass filter focusing on the lower part of the spectrum, we choose $s=10$, and the kernel function is chosen as a skewed bell-shape function obtained from the toolbox available at~\cite{SGWmex}. For the RF classifier $n_{tree}=100$, and the number of samples per node is set as: $f'=\floor{\sqrt{f}}$, where $f$ denotes the number of features. 
    
	The results are compared to the RBD algorithm \cite{15:Milling:InfoCom}. In the simulation, the radius of the ball and the threshold parameters of this algorithm are both obtained empirically so that the RBD algorithm exhibits the best performance for the given dataset. We present four different scenarios below, followed by a detailed discussion. 
	
% 		\begin{figure*}[t]
% 		\centering
%         \subcaptionbox{The GSL algorithm using the entire spectrum (E, $f=999$), the smoothness and the energy concentration metric.}[.3\linewidth][c]{%
% 		\includegraphics[width=.318\linewidth]{GSP_ER_100_4.pdf}}\quad
% 		\subcaptionbox{Graph wavelet-based algorithm with DBGW design.}[.3\linewidth][c]{%
% 		\includegraphics[width=.318\linewidth]{DBGW_ER_100_4.pdf}}\quad
% 		\subcaptionbox{Graph wavelet-based algorithm with SGW design.}[.3\linewidth][c]{%
% 		\includegraphics[width=.318\linewidth]{SGWT_ER_100_4.pdf}}
       
% 		\caption{Comparison between all the proposed methods and the RBD algorithm for ER graph with size $1000$ and parameter $\rho=0.01$, where the epidemic stems from $4$ initial seeds.}
% 		 \label{diag:scenario1}
% 	\end{figure*}
\begin{figure*}[t]
	\centering
			\subcaptionbox{Comparison between the metric-based algorithms (smoothness and energy concentration ratio) and the RBD algorithm.}[.3\linewidth][c]{%
			\includegraphics[width=.318\linewidth]{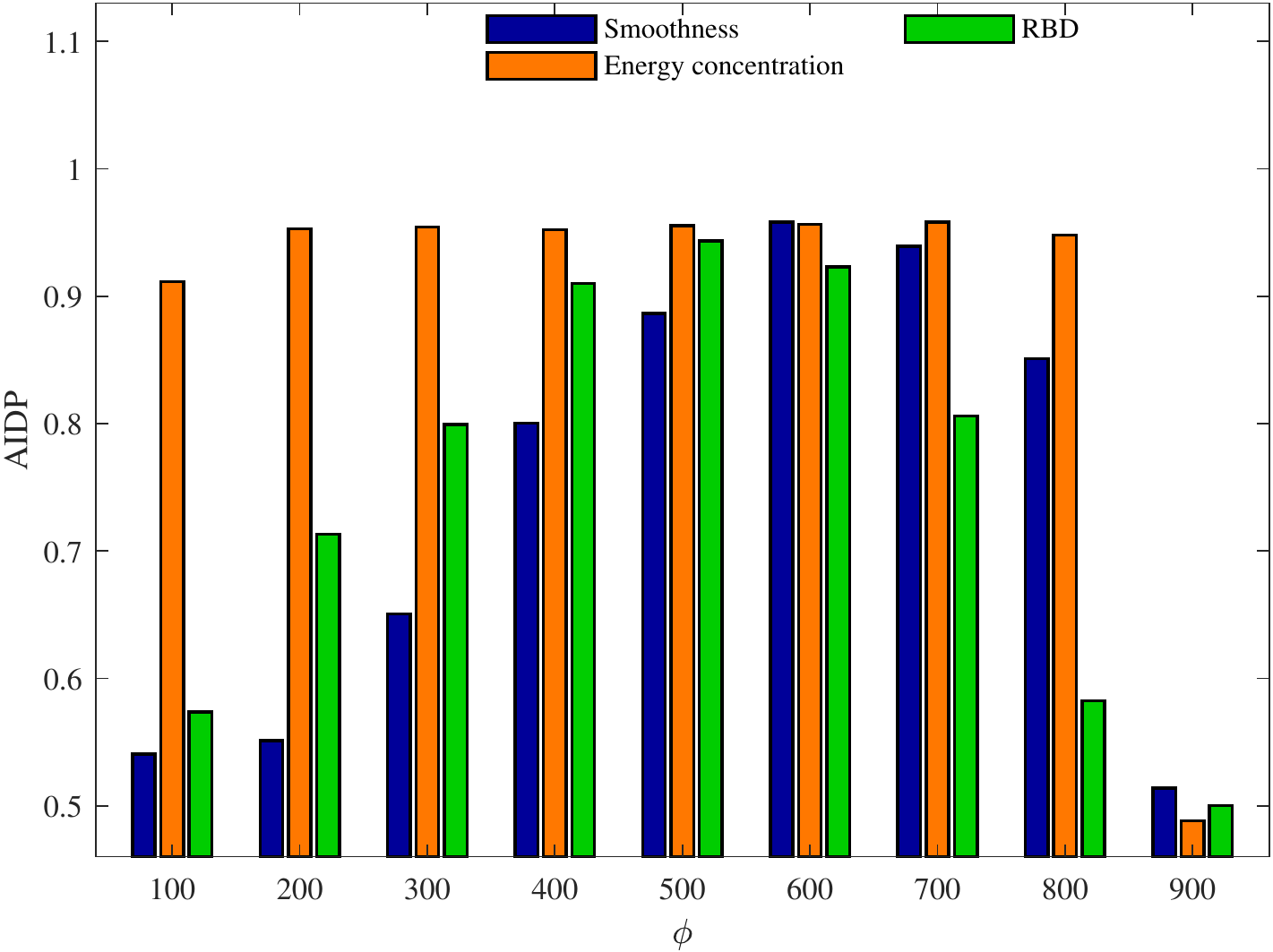}}\quad
			\subcaptionbox{Comparison between the machine learning-based algorithms (with features from GFT, SGW and DBGW)- direct feature feeding and the RBD algorithm.}[.3\linewidth][c]{%
			\includegraphics[width=.318\linewidth]{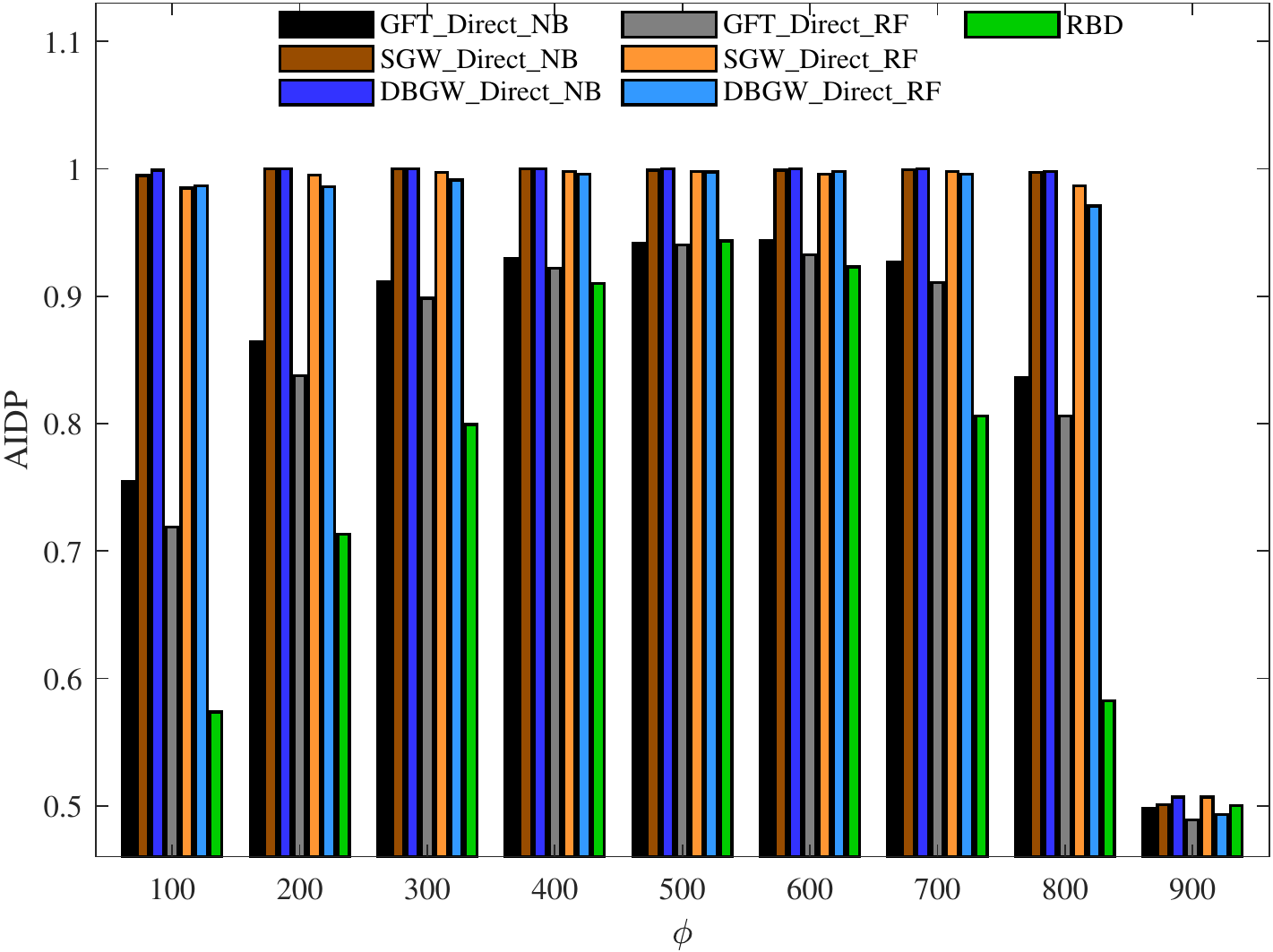}}\quad
			\subcaptionbox{Comparison between the machine learning-based algorithms (with features from GFT, SGW and DBGW)- fast feature feeding and the RBD algorithm.}[.3\linewidth][c]{%
			\includegraphics[width=.318\linewidth]{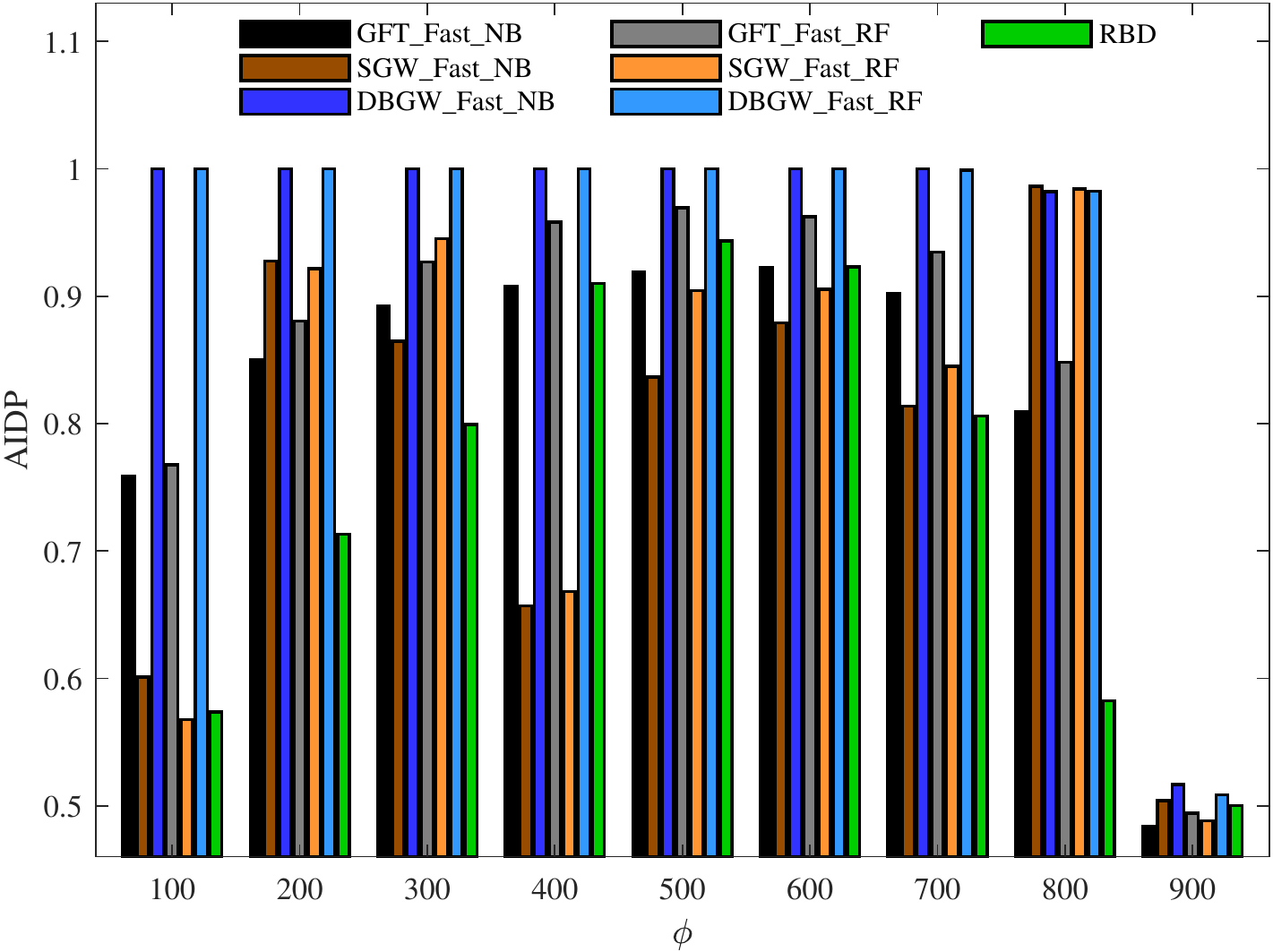}}
            
	\caption{Performance comparison between all the proposed methods and the RBD algorithm for ER graph with size $1000$ and parameter $\rho=0.01$, where the epidemic stems from $1$ initial seed.}
    \label{diag:scenario3}
\end{figure*}

	\begin{figure*}[t]
		\centering
        \subcaptionbox{Comparison between the metric-based algorithms (smoothness and energy concentration ratio) and the RBD algorithm.}[.3\linewidth][c]{%
		\includegraphics[width=.318\linewidth]{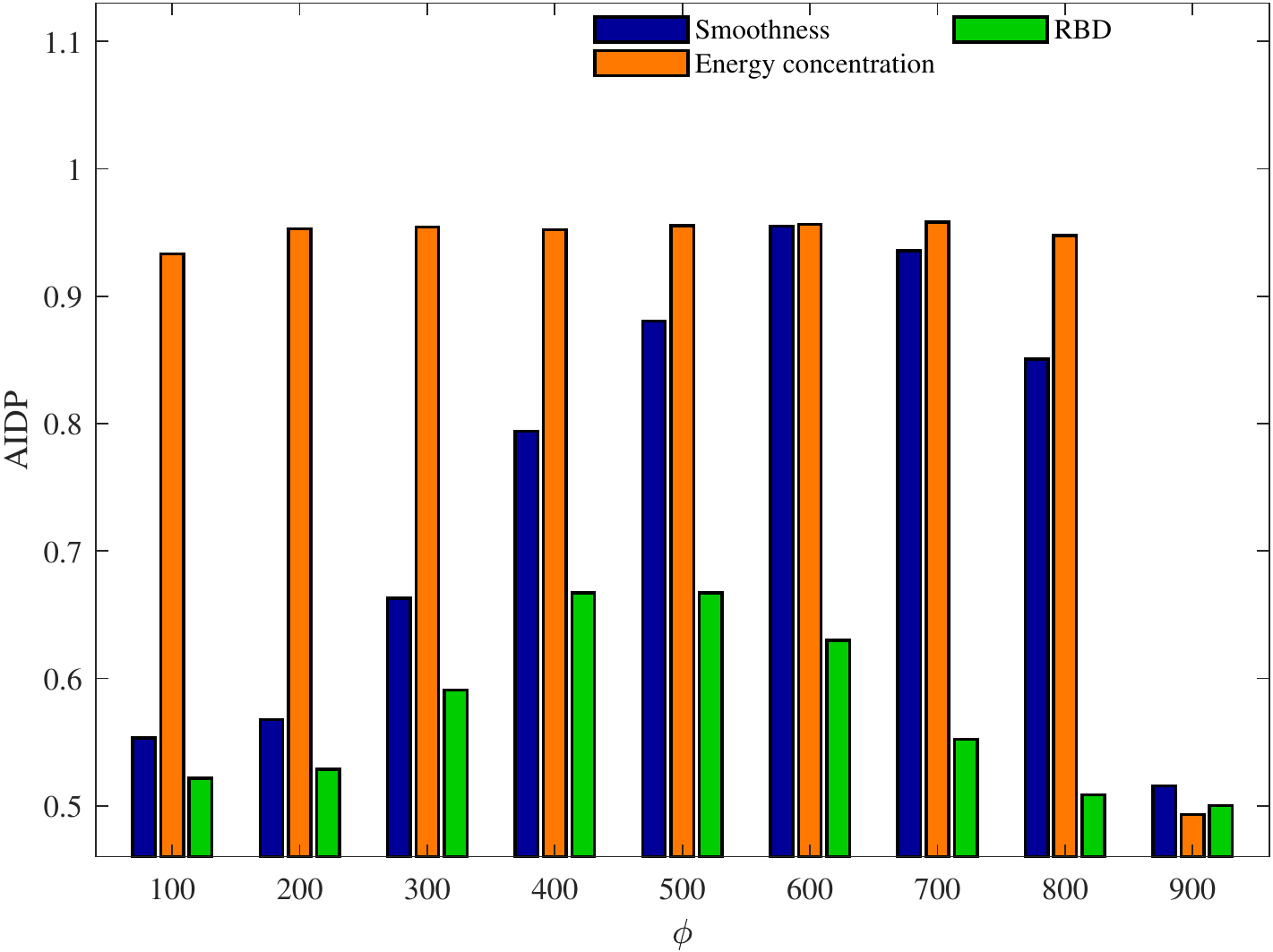}}\quad
		\subcaptionbox{Comparison between the machine learning-based algorithms (with features from GFT, SGW and DBGW)- direct feature feeding and the RBD algorithm.}[.3\linewidth][c]{%
		\includegraphics[width=.318\linewidth]{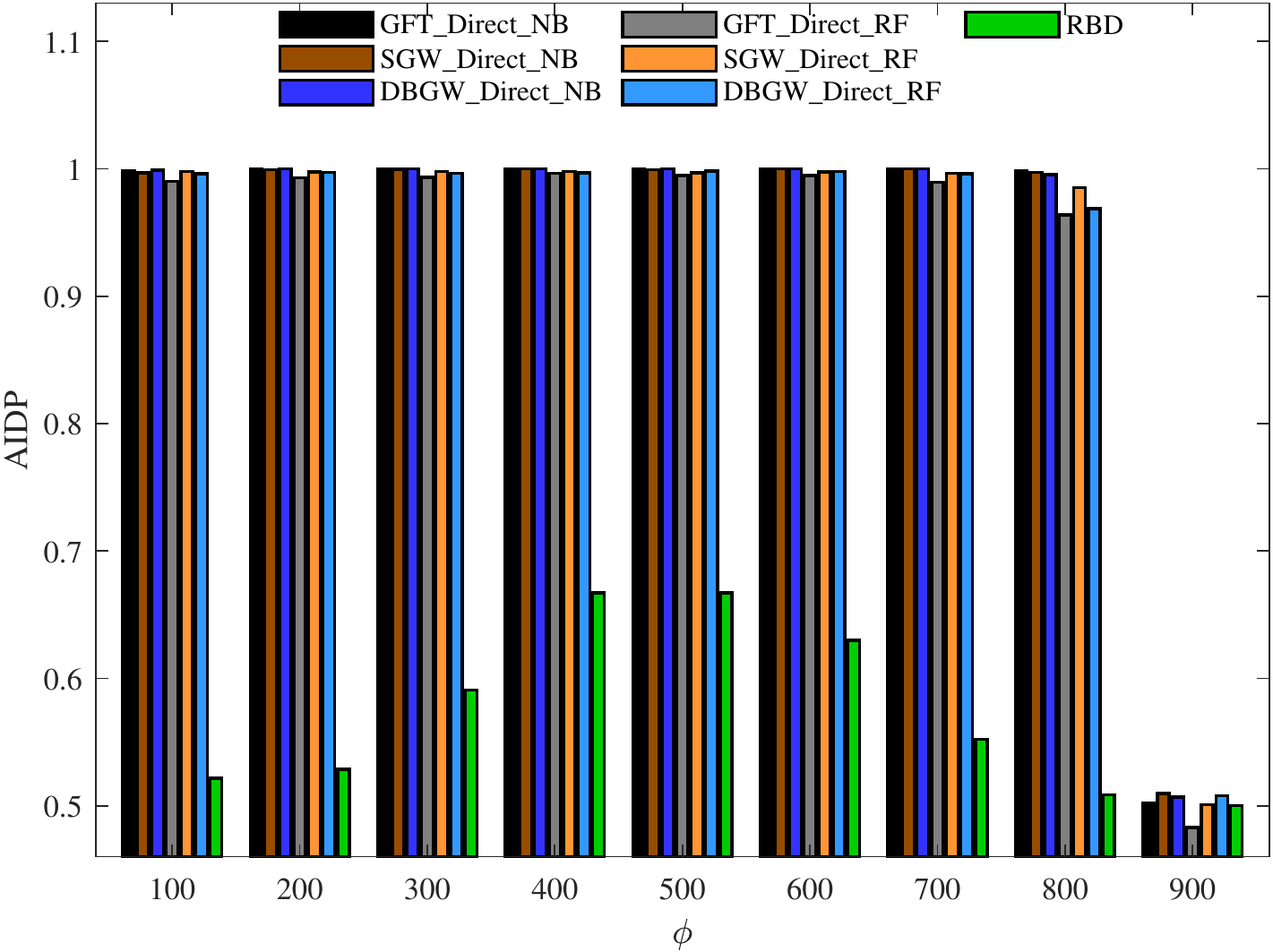}}\quad
		\subcaptionbox{Comparison between the machine learning-based algorithms (with features from GFT, SGW and DBGW)- fast feature feeding and the RBD algorithm.}[.3\linewidth][c]{%
		\includegraphics[width=.318\linewidth]{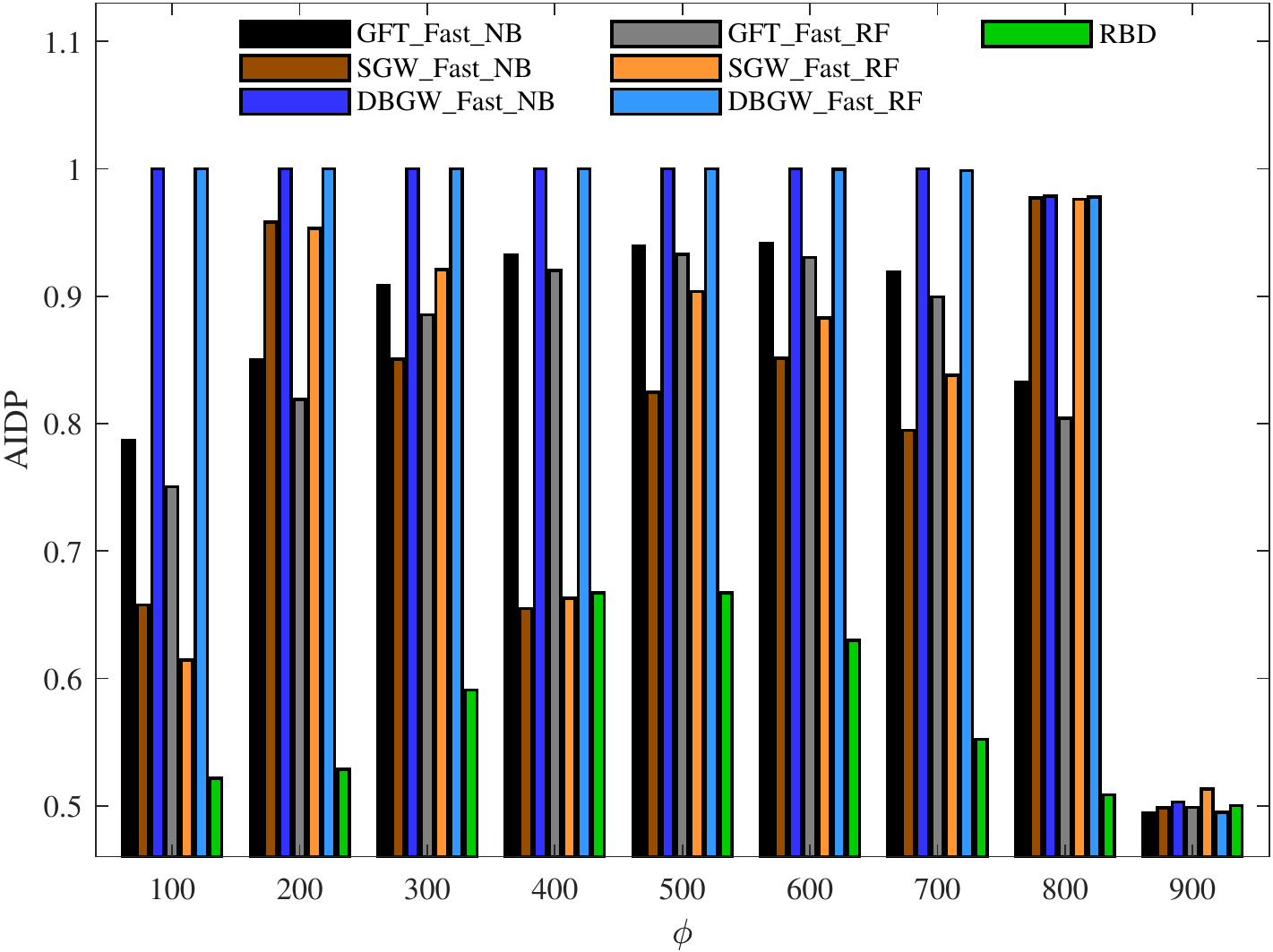}}
       
		\caption{Performance comparison between all the proposed methods and the RBD algorithm for ER graph with size $1000$ and parameter $\rho=0.01$, where the epidemic stems from $4$ initial seeds.}
		 \label{diag:scenario1}
	\end{figure*}
	%%%%%%%%%%%%%%%%%%%%%%
%%%%%%%%%%%%%%%%%%%%%%%
%%%%%%%%%%%%%%%%%%%%%%%%%%	%%%%%%%%%%%%%%%%%%%%%%
%%%%%%%%%%%%%%%%%%%%%%%
%%%%%%%%%%%%%%%%%%%%%%%%%%
% 	\begin{figure*}[t]
% 		\centering
%         \subcaptionbox{The GSL algorithm using the entire spectrum (E, $f=999$), the smoothness and the energy concentration metric.}[.3\linewidth][c]{%
% 		\includegraphics[width=.318\linewidth]{GSP_SF_3_4.pdf}}\quad
% 		\subcaptionbox{Graph wavelet-based algorithm with DBGW design.}[.3\linewidth][c]{%
% 		\includegraphics[width=.318\linewidth]{DBGW_SF_3_4.pdf}}\quad
% 		\subcaptionbox{Graph wavelet-based algorithm with SGW design.}[.3\linewidth][c]{%
% 		\includegraphics[width=.318\linewidth]{SGWT_SF_3_4.pdf}}
       
% 		\caption{Comparison between all the proposed methods and the RBD algorithm for SF graph with size $1000$ and parameter $\Gamma=3$, where the epidemic stems from $4$ initial seeds.}
%          \label{diag:scenario2}
% 	\end{figure*}

\begin{figure*}[t]
		\centering
        \subcaptionbox{Comparison between the metric-based algorithms (smoothness and energy concentration ratio) and the RBD algorithm.}[.3\linewidth][c]{%
		\includegraphics[width=.318\linewidth]{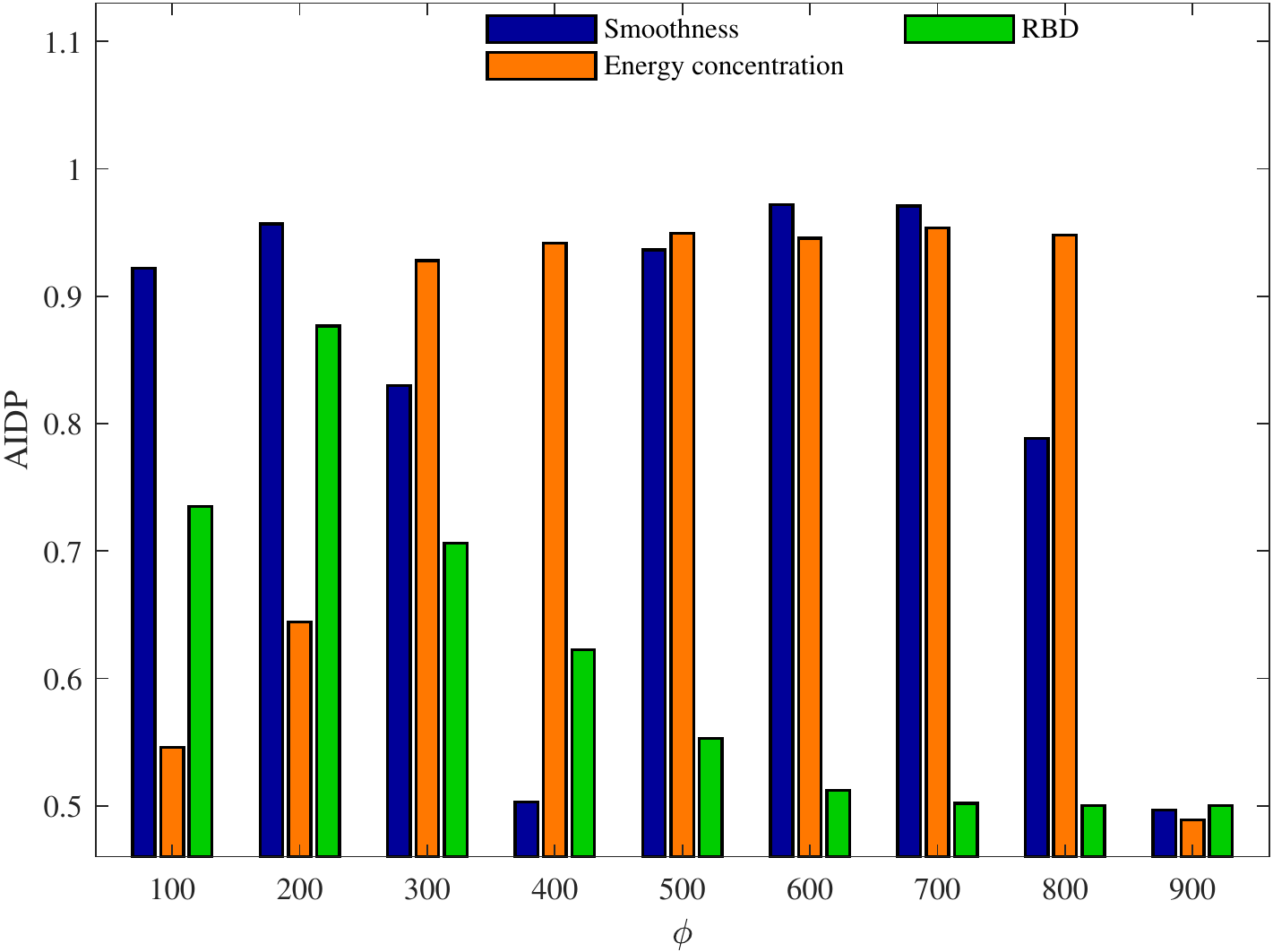}}\quad
		\subcaptionbox{Comparison between the machine learning-based algorithms (with features from GFT, SGW and DBGW)- direct feature feeding and the RBD algorithm.}[.3\linewidth][c]{%
		\includegraphics[width=.318\linewidth]{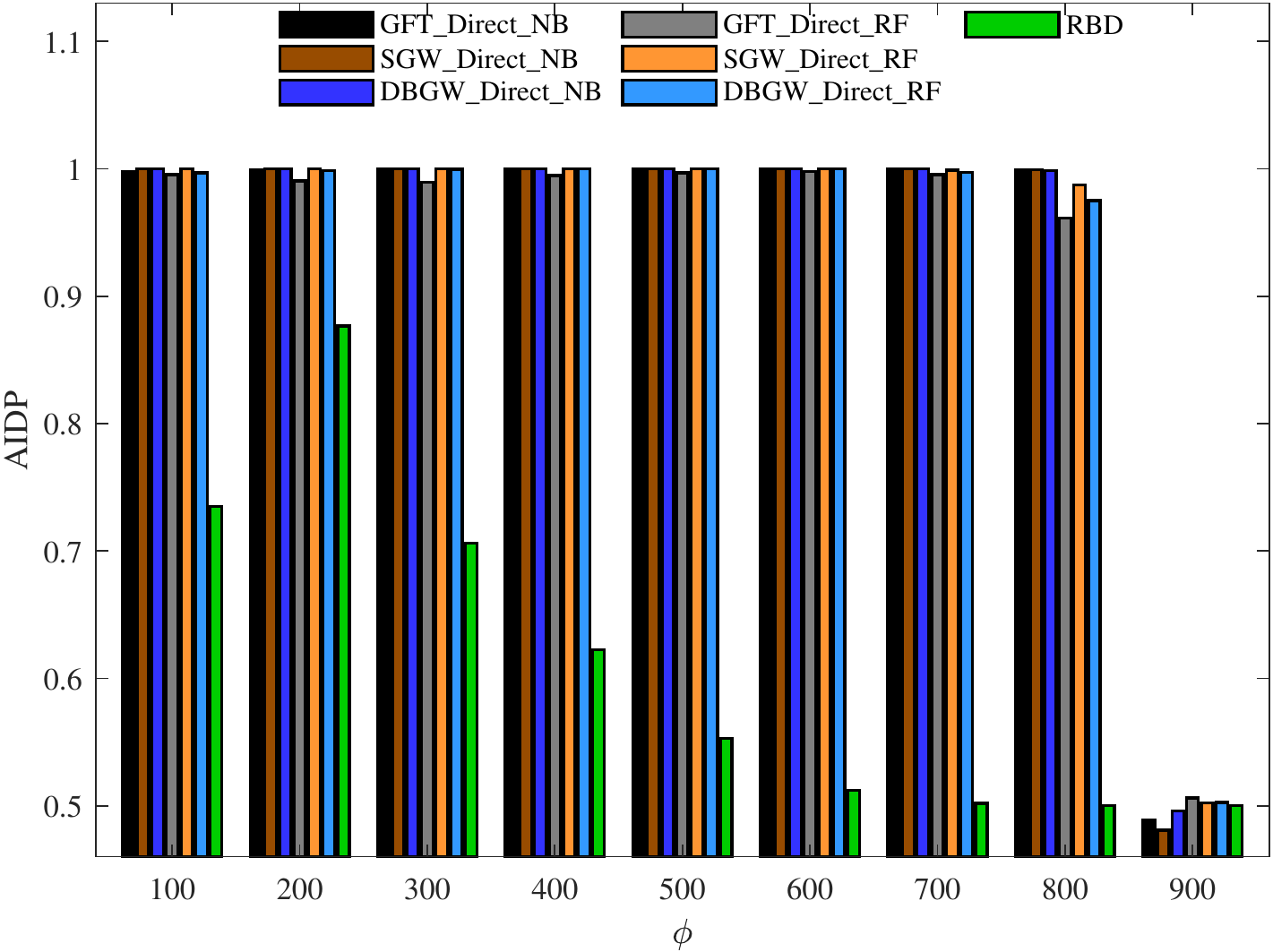}}\quad
		\subcaptionbox{Comparison between the machine learning-based algorithms (with features from GFT, SGW and DBGW)- fast feature feeding and the RBD algorithm.}[.3\linewidth][c]{%
		\includegraphics[width=.318\linewidth]{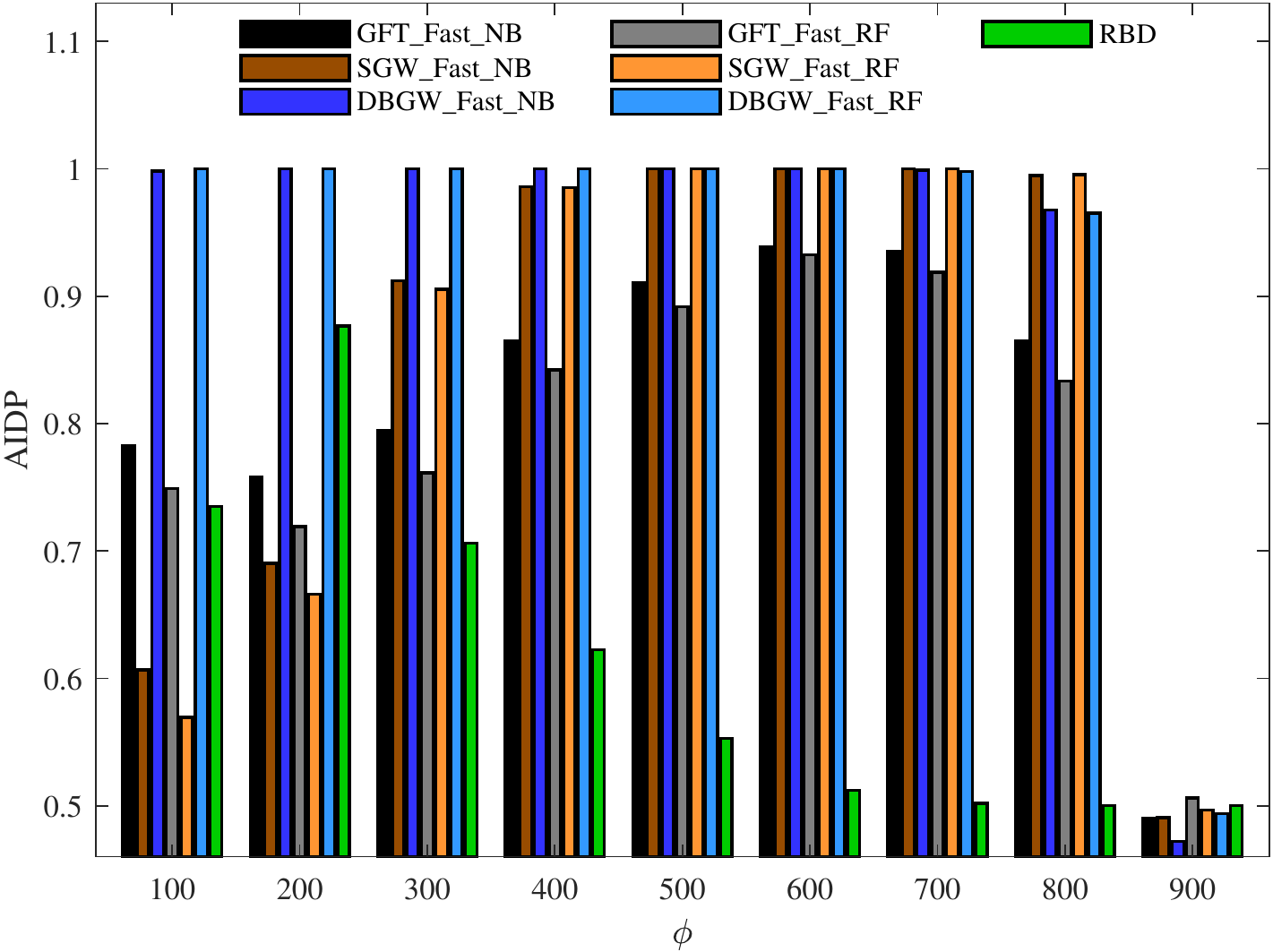}}
       
		\caption{Performance comparison between all the proposed methods and the RBD algorithm for SF graph with size $1000$ and parameter $\Gamma=3$, where the epidemic stems from $4$ initial seeds.}
         \label{diag:scenario2}
	\end{figure*}
	%%%%%%%%%%%%%%%%%%%%%%
%%%%%%%%%%%%%%%%%%%%%%%
%%%%%%%%%%%%%%%%%%%%%%%%%%	%%%%%%%%%%%%%%%%%%%%%%
% \begin{figure*}[t]
% 	\centering
% 			\subcaptionbox{The GSL algorithm using the entire spectrum (E, $f=999$), the smoothness and the energy concentration metric.}[.3\linewidth][c]{%
% 			\includegraphics[width=.318\linewidth]{GSP_Real.pdf}}\quad
% 			\subcaptionbox{Graph wavelet-based algorithm with DBGW design.}[.3\linewidth][c]{%
% 			\includegraphics[width=.318\linewidth]{DBGW_Real.pdf}}\quad
% 			\subcaptionbox{Graph wavelet-based algorithm with SGW design.}[.3\linewidth][c]{%
% 			\includegraphics[width=.318\linewidth]{SGW_Real.pdf}}
            
% 	\caption{Comparison between all the proposed methods and the RBD algorithm for the email-Eu-core network~\cite{snapnets}, where the epidemic stems from $4$ initial seeds.}
%     \label{diag:scenario4}
% \end{figure*}
\begin{figure*}[t]
	\centering
			\subcaptionbox{Comparison between the metric-based algorithms (smoothness and energy concentration ratio) and the RBD algorithm.}[.3\linewidth][c]{%
			\includegraphics[width=.318\linewidth]{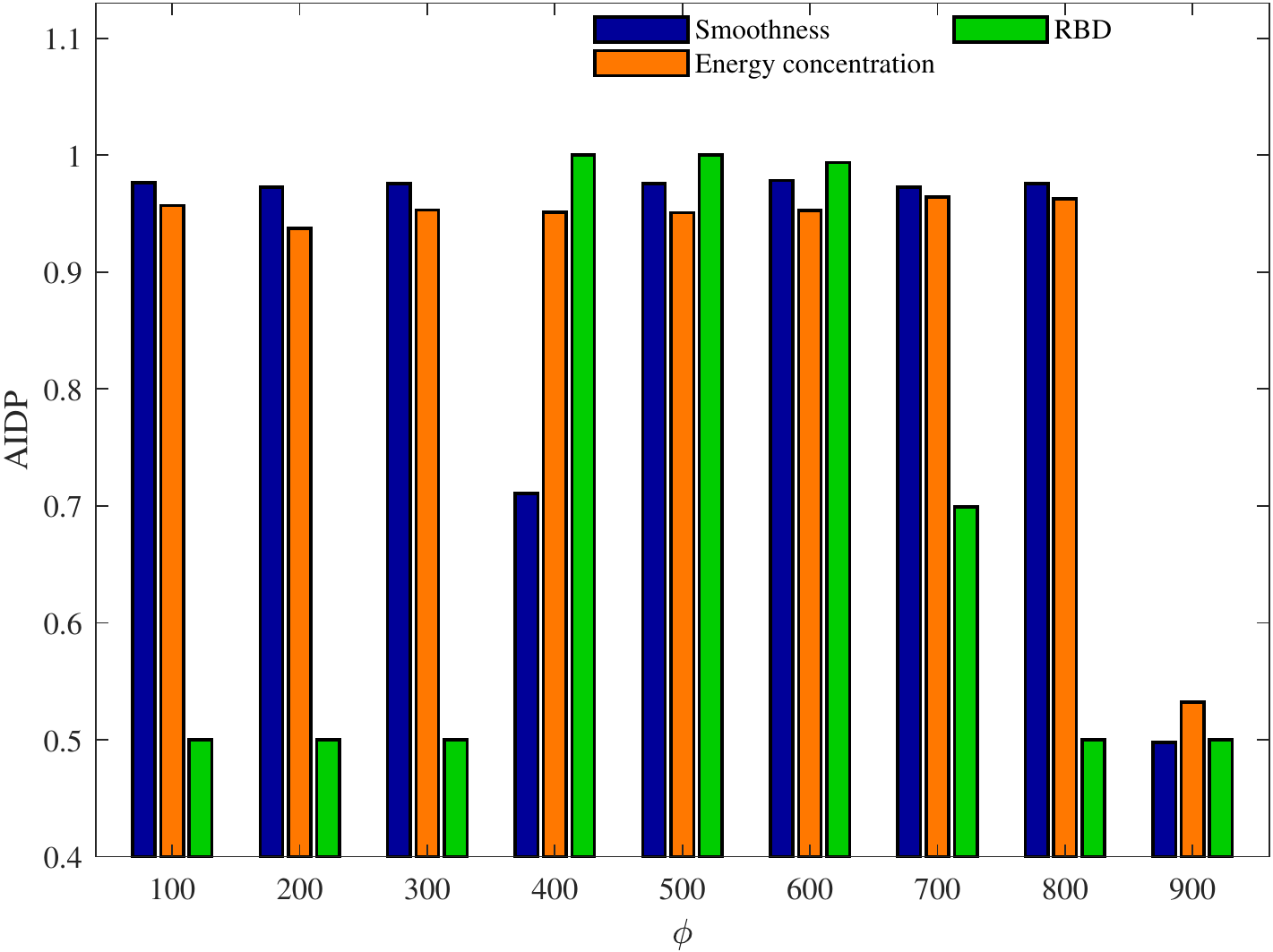}}\quad
			\subcaptionbox{Comparison between the machine learning-based algorithms (with features from GFT, SGW and DBGW)- direct feature feeding and the RBD algorithm.}[.3\linewidth][c]{%
			\includegraphics[width=.318\linewidth]{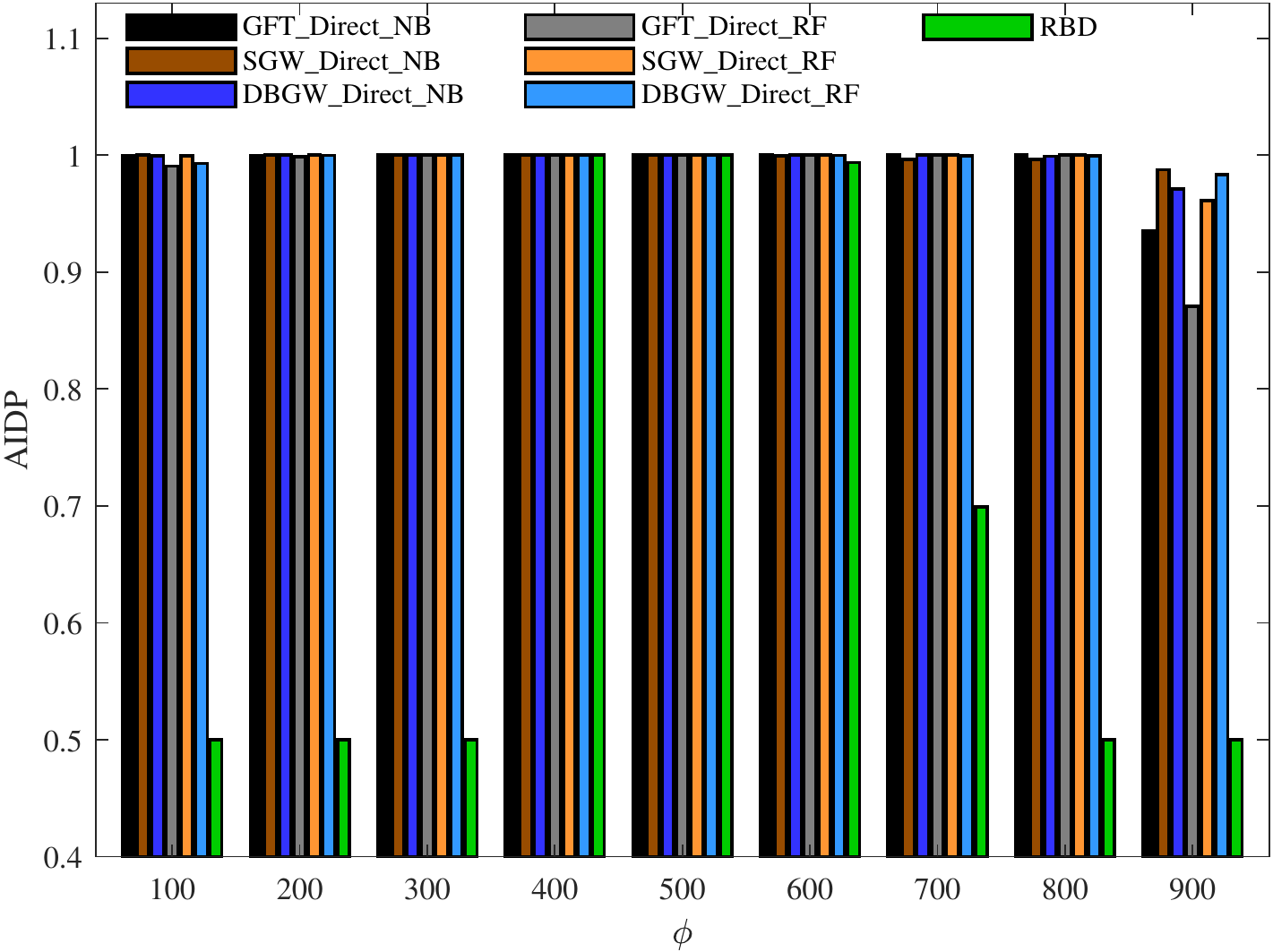}}\quad
			\subcaptionbox{Comparison between the machine learning-based algorithms (with features from GFT, SGW and DBGW)- fast feature feeding and the RBD algorithm.}[.3\linewidth][c]{%
			\includegraphics[width=.318\linewidth]{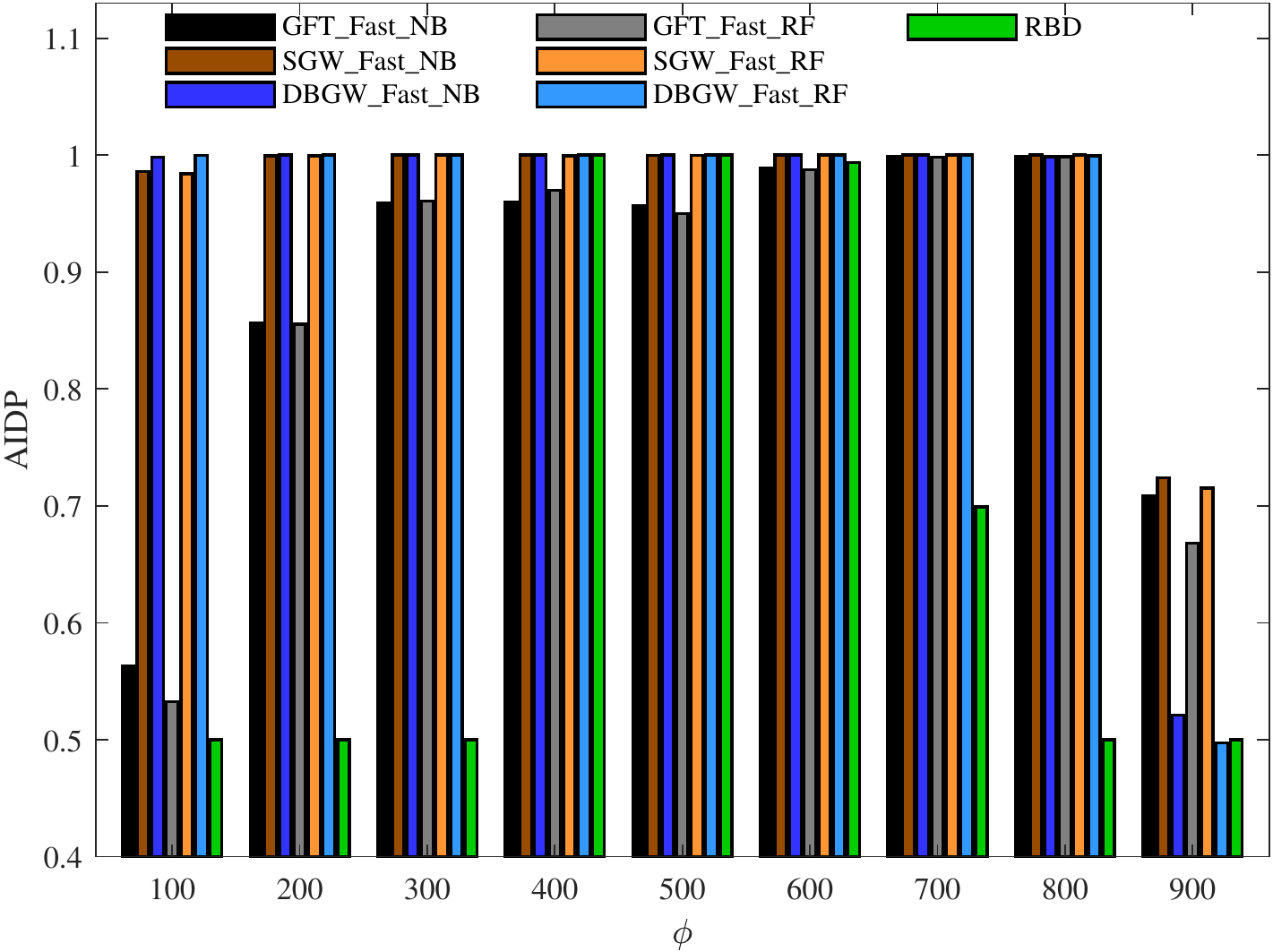}}
            
	\caption{Performance comparison between all the proposed methods and the RBD algorithm for the email-Eu-core network~\cite{snapnets}, where the epidemic stems from $4$ initial seeds.}
    \label{diag:scenario4}
\end{figure*}
%%%%%%%%%%%%%%%%%%%%%%%
%%%%%%%%%%%%%%%%%%%%%%%%%%
% 	\begin{figure*}[t]
% 	\centering
% 			\subcaptionbox{The GSL algorithm using the entire spectrum (E, $f=999$), the smoothness and the energy concentration metric.}[.3\linewidth][c]{%
% 			\includegraphics[width=.318\linewidth]{GSP_ER_100_1.pdf}}\quad
% 			\subcaptionbox{Graph wavelet-based algorithm with DBGW design.}[.3\linewidth][c]{%
% 			\includegraphics[width=.318\linewidth]{DBGW_ER_100_1.pdf}}\quad
% 			\subcaptionbox{Graph wavelet-based algorithm with SGW design.}[.3\linewidth][c]{%
% 			\includegraphics[width=.318\linewidth]{SGWT_ER_100_1.pdf}}
            
% 	\caption{Comparison between all the proposed methods and the RBD algorithm for ER graph with size $1000$ and parameter $\rho=0.01$, where the epidemic stems from $1$ initial seeds.}
%     \label{diag:scenario3}
% \end{figure*}

	\paragraph*{\textbf{Scenario 1}} We start with a scenario where the epidemic starts from one infected node, and the underlying graph structure is an ER graph with size $1000$ and $\rho = 0.01$. Simulation results for all the proposed methods are presented in Figure~\ref{diag:scenario3}, divided into three sub-figures for ease for illustration. In the first sub-figure, the metric-based infection analysis algorithms (c.f. Section~\ref{sec:MacroMetric}) are presented, and compared to the RBD algorithm. The other two sub-figures are concerned with the machine-learning based algorithms with features from graph Fourier spectrum and graph wavelets (c.f. Section~\ref{sec:MacroLeaning}), considering the direct and fast approach, respectively.
	
	\paragraph*{\textbf{Scenario 2}} In this scenario, the network structure is the same as the first scenario. In contrast, the epidemic stems from $4$ initial nodes chosen to be far away from each other to avoid merging of the epidemic in the initial stages of spreading; this is assumed for the next scenario as well. Simulation results for all the proposed methods are presented in Figure~\ref{diag:scenario1}.
	
	\paragraph*{\textbf{Scenario 3}} In this scenario, the underlying graph structure is an SF graph with size $1000$ and $\Gamma=3$. Simulation results for all the proposed methods are presented in Figure~\ref{diag:scenario2}.
    
    \paragraph*{\textbf{Scenario 4}} In this scenario, the underlying graph structure is the email-Eu-core network. The epidemic stems from $4$ randomly chosen nodes. Simulation results for all the proposed methods are presented in Figure~\ref{diag:scenario4}.
    
	\textbf{Discussion of the results:} As can be observed in Sub-fig 1 of Figure~\ref{diag:scenario3}-\ref{diag:scenario4}, in most of the cases, our metric based algorithms have a significantly better performance as compared to the RBD algorithm; in the cases that the RBD algorithm is marginally better (see Figure 6), our metric-based algorithms exhibit excellent detection performance nonetheless (above 93\% for $100\leq \phi \leq 800$). The machine learning based algorithms with features from graph Fourier spectrum and graph wavelets generally exhibit even better performance, as shown in Sub-fig 2 of Figures \ref{diag:scenario3}-\ref{diag:scenario4}. As we discussed earlier, the metric-based algorithms are generally simpler to implement, so there is certain tradeoff in selecting these two types of algorithms for infection analysis considering detection accuracy and simplicity of implementation. The unsatisfactory performance of the RBD algorithms is attributed to the irregular and dense structure of the underlying graphs considered here: according to their design mechanism, each ball would contain a large portion of nodes, and the density of infected nodes inside and outside of the ball would often be similar in such settings. Having multiple initial seeds of virus further deteriorates the RBD performance as expected. 
	
       In all these figures, the machine learning-based algorithm using DBGW exhibits the best performance across the board. For the fast method of feature feeding (Sub-fig 3 of the figures), this algorithm retains its excellent performance; however, the performance of the other two machine learning based algorithms (using GFT and SGW) decreases, sometimes significantly, when switching from full feature feeding to fast feature feeding. This indicates that statistical information of the coefficients of the DBGW (using only three features) is sufficient for infection detection, which is particularly appealing for infection analysis in large-scale networks.
       
    One unique characteristic of the machine learning-based algorithm using GFT is that the learning can be conducted using a few GFT components directly. This is because each element of the GFT spectrum captures the signal variation across the whole network, which is, in general, not the case in graph wavelets since graph wavelet coefficients are centered at different nodes by definition. Figure~\ref{diag:spec_learn_vs_DBGW} examines the performance of the machine-learning based algorithms with GFT and DBGWs for the epidemic processes described in Scenario $2$ upon using a few number of features $f$ for the NB classifier. In particular, we consider the following choices of the features for machine learning: (1) the first three GFT spectrum components  $\Scale[0.8]{\left[\hat{{S}}(\lambda_1),\hat{{S}}(\lambda_2),\hat{{S}}(\lambda_3)\right]}$ (denoted by (L, $f=3$)), (2) the middle three $\Scale[0.8]{\left[\hat{{S}}(\lambda_{499}),\hat{{S}}(\lambda_{500}),\hat{{S}}(\lambda_{501})\right]}$ (denoted by (MID, $f=3$)), (3) the last three $\Scale[0.8]{\left[\hat{{S}}(\lambda_{997}),\hat{{S}}(\lambda_{998}),\hat{{S}}(\lambda_{999})\right]}$ (denoted by (H, $f=3$)), (4) a combination of the first and the last two $\Scale[0.8]{\left[\hat{{S}}(\lambda_{1}),\hat{{S}}(\lambda_{998}),\hat{{S}}(\lambda_{999})\right]}$ (denoted by (COMB1, $f=3$)), and (5) another sampling approach using $\Scale[0.8]{\left[\hat{{S}}(\lambda_{1}),\hat{{S}}(\lambda_{2}),\hat{{S}}(\lambda_{15}),\hat{{S}}(\lambda_{999})\right]}$ (denoted by (COMB2, $f=4$)). Note that the last three components are the fastest to obtain in a large-scale network using the algorithm in~\cite{ref:eig1}. As can be expected, the middle part of the spectrum does not lead to a good performance. This implies that analyzing the low and the high frequency parts of the GFT spectrum is essential to conduct infection detection. 
%     Also, as can be seen, using the first three components of the spectrum leads to a good detection result when a small portion of nodes are infected; however, the performance decreases as the number of infected nodes increases. In contrary, using the last three components of the spectrum leads to a better detection performance for $400\leq \phi \leq 800$. This phenomenon is justified in the following. For the case that the network is slightly affected by epidemics, there exists a dense component of failed nodes and a dense component of healthy nodes; however, this is not the case in random failures. Thus, low frequency part of the spectrum can perform well in this regime. However, when a large portion of the nodes are infected, there would be a dense area of infected nodes, where the healthy nodes are mostly surrounded by infected nodes; and thus healthy nodes are randomly scattered over the network. In this regime, the epidemics and random failures have to be distinguished by conducting high frequency analysis, which concerns the rapid signal variation among the neighboring nodes. 
Furthermore, the combined method leads to more robust performance. Also, by using one more sample (Comb2), one can get almost the same average performance as compared to the fast method ($2\%$ worse on average). This figure also confirms the superiority of using DBGWs as features for machine learning. In particular, for $100 \leq \phi \leq 800$, using fast DBGW has a performance gain between $7\%$ and $20\%$ as compared to the fast GFT.
    	\begin{figure}[h]
		\includegraphics[width=0.43\textwidth]{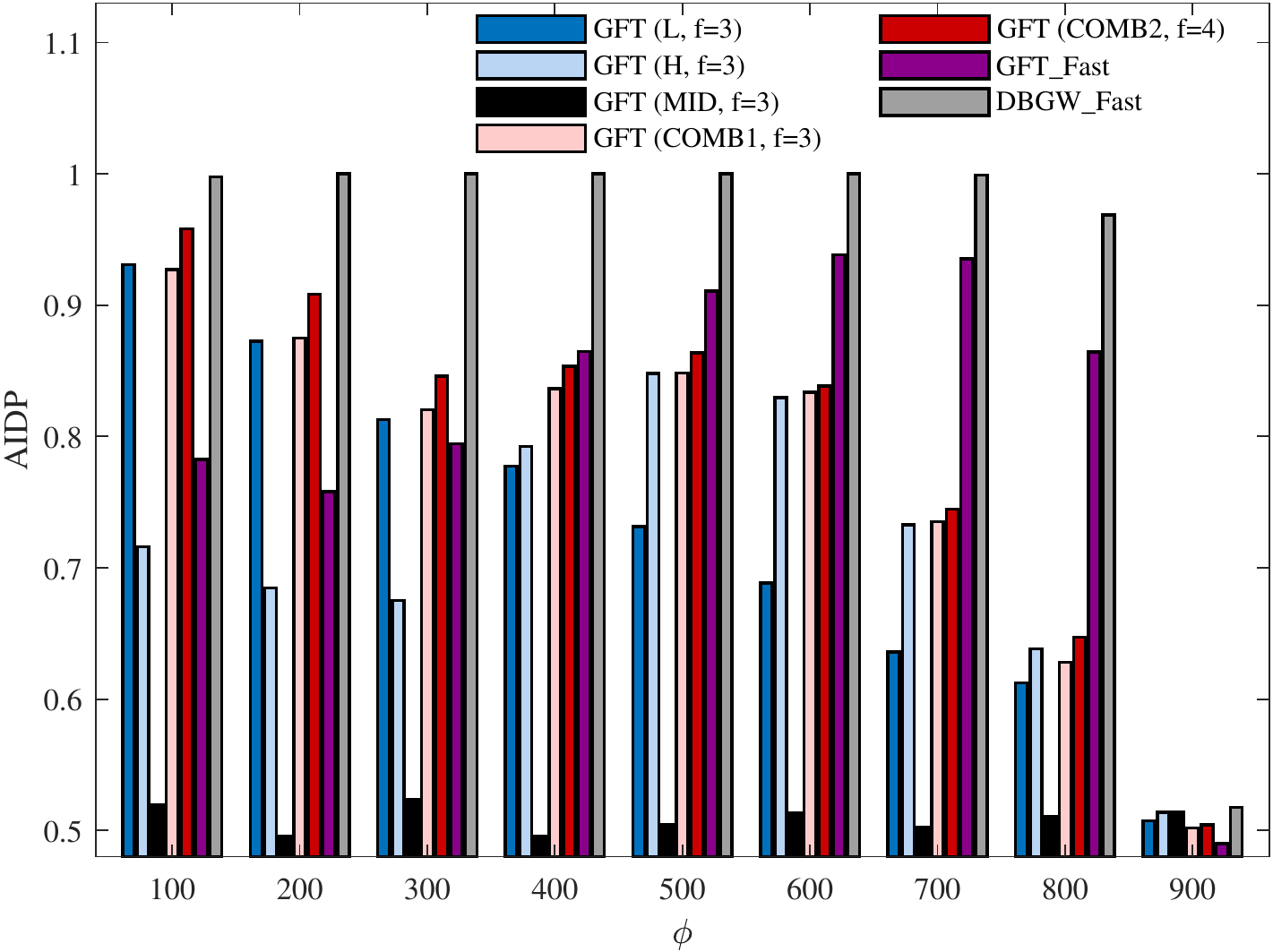}
		\centering
		\caption{Comparison between the machine learning-based algorithms with different number of GFT components and DBGWs with fast feature feeding using the NB classifier.}
		\label{diag:spec_learn_vs_DBGW}
	\end{figure}
    	\begin{figure*}[t]
		\includegraphics[width=7.1 in, height=1.4 in]{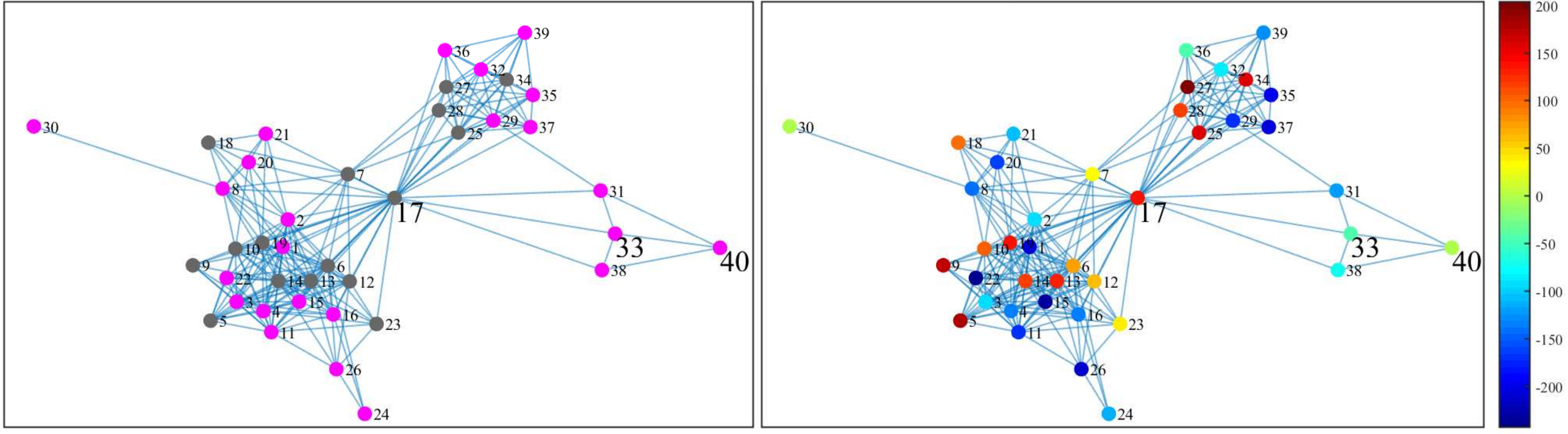}
		\centering
		\caption{The snapshot of an epidemic (left) and the corresponding nodes' DBGW coefficients (right). In the left plot, gray nodes are infected while the magenta nodes are healthy.}
		\label{diag:epidemicShow}
	\end{figure*}
		%%%%%%%%%%%%%%%%%%%%%%%%%%%%%%%%%%%%%%%%%%%%%%%%%%%
    %%%%%%%%%%%%%%%%%%%%%%%%%%%%%%%%%%%%%%%%%%%%%%%%%%%%%
    %%%%%%%%%%%%%%%%%%%%%%%%%%%%%%%%%%%%%%%%%%%%%%%%%%%%%%%
    %%%%%%%%%%%%%%%%%%%%%%%%%%%%%%%%%%%%%%%%%%%%%%%%%%%%%%
    \vspace{-1mm}
	\subsection{Micro Analysis of Infections}\label{sub:mi}
	We conduct micro infection analysis on the Random Geometric Graph (RGG)~\cite{RGG}, dataset107 and dataset698 networks described in Table~\ref{tabel:datasets}. The RGGs are widely used to model ad-hoc wireless networks~\cite{RGGApp}. An RGG is formed by placing $N$ nodes uniformly at random in the unit square and connecting adjacent nodes by an edge if the Euclidean distance between them is at most $\Xi$, where $N=1000, \Xi=0.08$ in our construction. These networks possess large diameters as compared to the SF and ER graphs with the same size and number of edges. The two latter networks are both sampled instances of the Facebook network~\cite{McAuley:2012}. Dataset698 with fewer nodes is used for the illustration purpose, while the experimental results are obtained on dataset107 and the RGG. For all networks, we process their graphs to eliminate island components prior to simulations. The network and DBGW settings are the same as~\ref{subsec:Macronumerical} except that $\beta=0.25$ to account for a slower spreading of the virus. 
%     As before, we have $A=1000$, $\beta=0.5$. Also, the scale of the DBGW $s$ is to $1$.
	\subsubsection{Interpretation of Coefficients of DBGWs}
	For an epidemic stemming form a single source node on the dataset698 network, Figure~\ref{diag:epidemicShow} depicts the network snapshot (left plot) and the corresponding nodes' DBGW coefficients (right plot), where $40\%$ of the nodes are infected. The result of Theorem~\ref{th:sign} can be verified since nodes with the same state (in the left plot) possess wavelet coefficients of the same sign (in the right plot); in particular, the infected nodes exhibit positive wavelet coefficients and the healthy nodes exhibit negative ones. To verify the results of Corollary~\ref{cor:amp} and Remark~\ref{rem6}, consider nodes $17,33,40$ in Figure~\ref{diag:epidemicShow}. Node $17$ is infected and has multiple neighbors in the opposite state; hence it has a large positive DBGW coefficient. Node $33$ is healthy and located in a neighborhood with only one node in the opposite state; hence it has a coefficient of small magnitude. Finally, node $40$ is surrounded by the nodes in same state; hence the DBGW coefficient is zero.\footnote{Due to the consideration of proximity of the neighbors reflected by edge weights in DBGWs, nodes with the same degree and the same number of neighbor nodes with an opposite state may have different coefficients.} The set of candidates for the purpose of quarantine/vaccination comprises the nodes with the largest/smallest DBGW's wavelet coefficients. In this case, a good set of candidates for quarantine comprises nodes $17,5,25,27$, which are surrounded with mostly healthy nodes. Similarly, nodes $11,35,37,29$ are good suggestions for the purpose of vaccination.

	\begin{table}
	\resizebox{\columnwidth}{!}{%
			\begin{tabular}{|c| c| c|c|c|c|c|} 
				\hline
				Statistics &ER& SF &email-Eu-core& dataset107& dataset698&RGG    \\ 
				\hline\hline
				Order & $1000$&$1000$&$986$&$1033$ & $40$  & $1000$  \\ 
				\hline
				Size & $4995$&$2987$&$16064$ &$26747$ & $220$ &  $9218$  \\
				\hline
				AD & $9.99$&$5.97$&$32.6$&$51.78$ & $11$ & $18.43$ \\ 
				\hline
				Diameter &$5$  &$6$&$7$&$9$ & $4$ & $21$  \\
				\hline
				ASDL & $3.26$ &$3.5$&$2.59$ &$2.95$ & $1.94$  & $8.18$ \\
				\hline
			\end{tabular}
    }
		\caption{Statistics of the networks. (Order: number of nodes, Size: number of edges, AD: average degree, ASDL: average shortest distance length)}
		\label{tabel:datasets}

	\end{table}

		\subsubsection{Experimental Results}
	Figure \ref{diag:quar} depicts the effectiveness of quarantine using the nodes' DBGW coefficients as a countermeasure to an epidemic simulated on the RGG (left sub-figure) and the dataset107 (right sub-figure). The height of each bar indicates the average time for the epidemic to reach a certain number of infected nodes ($\phi$), after quarantining $20\%$ of the $200$ infected nodes at the beginning of the observation ($t=1$) and quarantining $20\%$ of infected nodes each time the epidemic infects $100$ more healthy nodes. Therefore, a taller bar implies a more effective countermeasure. As a baseline, the epidemic evolution without deploying any countermeasure is depicted along with quarantining via random, and degree-based, (i.e., quarantining those infected nodes with the highest degrees). For each case, $100$ Monte Carlo simulations are conducted, and the error bars indicate the standard deviation of the results. As can be seen, DBGW-based selection proves to be the most effective quarantining method. This effect is especially preeminent in the RGG (on average $50\%$ longer time as compared to the degree-based) due to its large diameter. For the dataset107 network, due to its dense structure and relatively short diameter, at $\phi = 300$ there are still plenty of healthy nodes in the ``core'', so the non-quarantined infected nodes can still infect many victims, which makes the quarantine ineffective. Nonetheless, the performance gap between the DBGW-based method and others increases with $\phi$, and reaches $50\%$ compared to the degree-based method at $\phi=500$.

	\begin{figure}[t]
		\includegraphics[width=3.32in,height=1.75in]{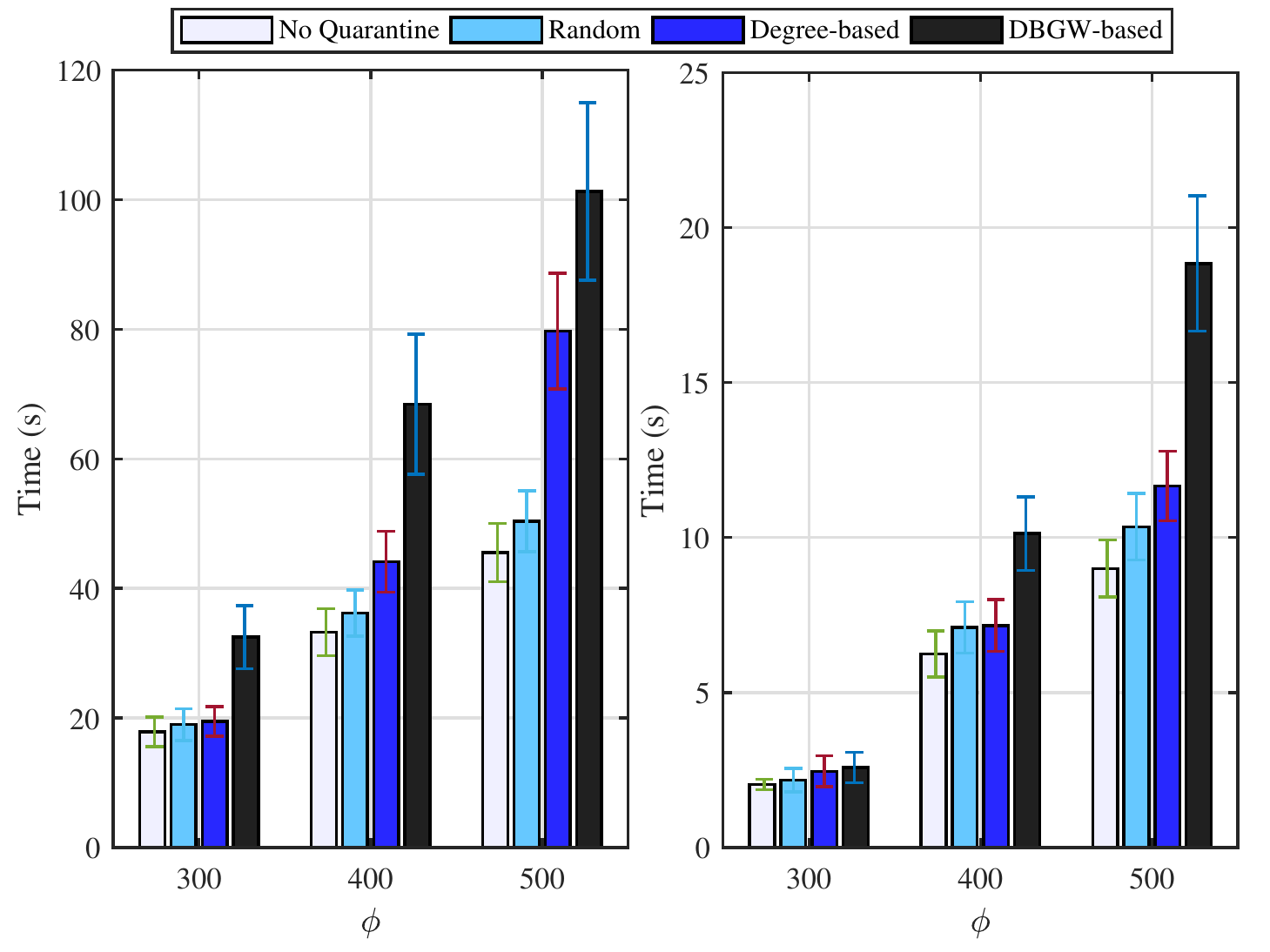}
		\centering
		\caption{Comparison between different methods of quarantine on the RGG (left) and the dataset107 (right) described in Table~\ref{tabel:datasets}.}
		\label{diag:quar}
	\end{figure}
    	%%%%%%%%%%%%%%%%%%%%%%%%%%%%%%%%%%%%%%%%%%%%%%%%%%%
    %%%%%%%%%%%%%%%%%%%%%%%%%%%%%%%%%%%%%%%%%%%%%%%%%%%%%
    %%%%%%%%%%%%%%%%%%%%%%%%%%%%%%%%%%%%%%%%%%%%%%%%%%%%%%%
    %%%%%%%%%%%%%%%%%%%%%%%%%%%%%%%%%%%%%%%%%%%%%%%%%%%%%%
	\section{Conclusion and Future Work}\label{sec:conclusion}
	\noindent	In this paper, we conducted both the macro and the micro infection analysis for irregular and heterogeneous networks, through graph signal processing techniques. For the macro analysis: 1) effective metrics based on the GFT spectrum are proposed; 2) a new class of graph wavelets, DBGWs, are introduced; and 3) both metric-based and machine learning-based infection analysis algorithms are developed. For the micro analysis, new vaccination and quarantine countermeasures using DBGWs are proposed. Through extensive simulations, the superiority of all the proposed algorithms are revealed as compared to the state-of-the-art. Given the simplicity and effectiveness of the proposed approaches in this paper, we expect that this work can shed light on further applications of graph signal processing on networking problems. One interesting future direction is to extend the study to time varying and multilayer networks. 
    \appendices
\section{Robustness Analysis of the Metrics}\label{app:A}   
Upon having a large number of faulty observations, distinguishing an epidemic from random failures can become impossible. In the following, for each metric, upper-bounds on the number of faulty observations are derived to ensure the correct detection. Also, the results are derived based on a given number of false positives $n_{fp}$ and false negatives $n_{fn}$, where $n_f=n_{fp}+n_{fn}$. In this case, the noise vector (described in~\ref{subsec:noise}, $\mathbf{n}$, takes the value of $2A$ in $n_{fp}$ indices, the value of $-2A$ in $n_{fn}$ indices, and the value of zero elsewhere.

	\begin{theorem}\label{th:HECRerror}
		For the energy concentration high metric, given an $\epsilon \%$-prediction interval for random failures $[C_s,C_e]$, if $HECR_{\alpha} (\hat{\mathbf{S}}^{(i)})=C_i$, the following bounds guarantee the correct detection upon existence of an epidemic: 
		
		\textbf{Case 1: If $C_i<C_s$:}
		\begin{flalign}
		n_{f} < \bigg[ \frac{ C_s \widetilde{Eng}(\hat{\mathbf{S}}^{(n)}) -\sum_{j = \floor{N(1-\alpha)}}^{N-1} \big|\widetilde{\hat{{S}}^{(i)}(\lambda_j)}\big|}{2  (1+\floor{N\alpha})\max_{j \in \big[\floor{N(1-\alpha)}, N-1\big]} \norm{\mathbf{u}_j}_\infty}\bigg]^+.
		\end{flalign}
		
		\textbf{Case 2: If $C_i>C_e$:}	
		
		\begin{flalign}
		\begin{split}
		&n_{f} < \bigg[ 	\frac{-\sum_{j = \floor{N(1-\alpha)}}^{N-1} \widetilde{\hat{\mathbf{S}}^{(i)}(\lambda_j)} - C_e \widetilde{Eng}(\hat{\mathbf{S}}^{(n)}) }{2(1+\floor{N\alpha})\max_{j \in \big[\floor{N(1-\alpha)}, N-1\big]}\norm{\mathbf{u}_j}_\infty}\bigg]^+,\\
        &\textrm{if  }0>\sum_{j = \floor{N(1-\alpha)}}^{N-1} <\mathbf{S}^{(n)},\mathbf{u}_j>\\
		&n_{f} <\bigg[ \frac{\sum_{j = \floor{N(1-\alpha)}}^{N-1} \widetilde{\hat{\mathbf{S}}^{(i)}(\lambda_j)} -  C_e\widetilde{Eng}(\hat{\mathbf{S}}^{(n)})}{2 (1+\floor{N\alpha})\max_{j \in \big[\floor{N(1-\alpha)}, N-1\big]}\norm{\mathbf{u}_j}_\infty}\bigg]^+,\\
       & \textrm{if  }0\leq \sum_{j = \floor{N(1-\alpha)}}^{N-1} <\mathbf{S}^{(n)},\mathbf{u}_j>,
		\end{split}
		\end{flalign}
        where $[a]^+=\max(a,0)$ for $a \in \mathbb{R}$.
	\end{theorem}
    \begin{proof}
	 For Case 1, $HECR_{\alpha}(\hat{\mathbf{S}}^{(n)}) < C_s$ ensures the correct detection. 
		This implies:
		\begin{align}\label{eq:midproof500}
		&HECR_{\alpha}(\hat{\mathbf{S}}^{(n)})=\frac{\sum_{j = \floor{N(1-\alpha)}}^{N-1} |\hat{\mathbf{S}}^{(n)}(\lambda_j)|}{Eng(\hat{\mathbf{S}}^{(n)})}\nonumber \\
		&= \frac{\sum_{j = \floor{N(1-\alpha)}}^{N-1}\abs[\Big]{ <\mathbf{S}^{(i)},\mathbf{u}_j> + <\mathbf{n},\mathbf{u}_j>}}{Eng(\hat{\mathbf{S}}^{(n)})}<C_s.
		\end{align}
		In a Hilbert space, any two vectors $a,b$ satisfy the triangle inequality, $||a|-|b|| \leq|a+b| \leq |a|+|b|$. Hence, using the maximum of the numerator, we get:
		\begin{equation}\label{eq:midproof1}
		\begin{aligned}
		&\frac{\sum_{j = \floor{N(1-\alpha)}}^{N-1}\abs[\Big]{ <\mathbf{S}^{(i)},\mathbf{u}_j> + <\mathbf{n},\mathbf{u}_j>}}{Eng(\hat{\mathbf{S}}^{(n)})} < \\
		&\frac{\sum_{j = \floor{N(1-\alpha)}}^{N-1}| <\mathbf{S}^{(i)},\mathbf{u}_j>| + \abs[\big]{\sum_{k=1}^{N} \mathbf{n}(k)\mathbf{u}_j(k)}}{Eng(\hat{\mathbf{S}}^{(n)})} .
		\end{aligned}
		\end{equation}
		Using the triangle inequality, we get:
		\begin{align}\label{eq:nfbound}
		&\Big|\sum_{k=1}^{N} \mathbf{n}(k)\mathbf{u}_j(k) \Big| \leq 2A n_{fp}\norm{\mathbf{u}_j}_\infty+2A n_{fn}\norm{\mathbf{u}_j}_\infty \nonumber \\
		&=2A n_{f}\norm{\mathbf{u}_j}_\infty.
		\end{align}
		Using the above inequality in Eq.~\eqref{eq:midproof1} and the condition for correct detection in Eq.~\eqref{eq:midproof500} yields:
        
		\begin{equation}
		n_{f} < \frac{ C_s Eng(\hat{\mathbf{S}}^{(n)}) -\sum_{j = \floor{N(1-\alpha)}}^{N-1}| <\mathbf{S}^{(i)},\mathbf{u}_j>|}{2A (1+\floor{N\alpha}) \max_{j \in \big[\floor{N(1-\alpha)}, N-1\big]} \norm{\mathbf{u}_j}_\infty},
		\end{equation}
		which leads to the theorem result.
		For Case 2, we need to have $HECR_{\alpha}(\hat{\mathbf{S}}^{(n)}) > C_e$. For a sequence of real numbers $\{a_i\}_{i=1}^p$, we can get: $\sum_{i=1}^p |a_i| \geq \big|\sum_{i=1}^p a_i\big|$; hence:
		\begin{align}
		&\sum_{j = \floor{N(1-\alpha)}}^{N-1}\abs[\Big]{ <\mathbf{S}^{(i)},\mathbf{u}_j> +\sum_{k=1}^{N} \mathbf{n}(k)\mathbf{u}_j(k)} \nonumber \\
		&> \abs[\Bigg]{\sum_{j = \floor{N(1-\alpha)}}^{N-1} <\mathbf{S}^{(i)},\mathbf{u}_j> +\sum_{k=1}^{N} \mathbf{n}(k)\mathbf{u}_j(k)} \nonumber \\
		&>Eng(\hat{\mathbf{S}}^{(n)}) C_e,
		\end{align}
		which leads to the two following conditions:
		\begin{equation}\label{eq:midCondi}
		\begin{aligned}
		&\sum_{j = \floor{N(1-\alpha)}}^{N-1} <\mathbf{S}^{(i)},\mathbf{u}_j> +\sum_{k=1}^{N} \mathbf{n}(k)\mathbf{u}_j(k)  \\
		&> Eng(\hat{\mathbf{S}}^{(n)}) C_e,\;\textrm{if }0\leq\sum_{j = \floor{N(1-\alpha)}}^{N-1} <\mathbf{S}^{(n)},\mathbf{u}_j>,  \\
        \end{aligned}
		\end{equation}
		%& \textrm{or}  \\
        \begin{equation}\label{eq:midCondi2}
		\begin{aligned}
		&\sum_{j = \floor{N(1-\alpha)}}^{N-1} <\mathbf{S}^{(i)},\mathbf{u}_j> +\sum_{k=1}^{N} \mathbf{n}(k)\mathbf{u}_j(k)  \\
		&<-Eng(\hat{\mathbf{S}}^{(n)}) C_e,\;\textrm{if } 0>\sum_{j = \floor{N(1-\alpha)}}^{N-1} <\mathbf{S}^{(n)},\mathbf{u}_j>.
		\end{aligned}
		\end{equation}

		Similar to Eq.~\eqref{eq:nfbound}, it can be derived that:
		\begin{align}
	\sum_{k=1}^{N} \mathbf{n}(k)\mathbf{u}_j(k) \geq 
		-2A n_{f}\norm{\mathbf{u}_j}_\infty.
		\end{align}
		
		Using this fact and the condition in Eq.~\eqref{eq:midCondi}, we get:
		\begin{equation}
		\begin{aligned}
		\sum_{j = \floor{N(1-\alpha)}}^{N-1} \hat{{S}}^{(i)}(\lambda_j) -2A n_{f}\norm{\mathbf{u}_j}_\infty  
		>Eng(\hat{\mathbf{S}}^{(n)}) C_e .
        \end{aligned}
		\end{equation}
        Therefore,
        \begin{equation}
		\begin{aligned}
		n_f< \frac{\sum_{j = \floor{N(1-\alpha)}}^{N} \hat{{S}}^{(i)}(\lambda_j) - Eng(\hat{\mathbf{S}}^{(n)}) C_e}{2A(1+ \floor{N\alpha})\max_{j \in \big[\floor{N(1-\alpha)}, N-1\big]}\norm{\mathbf{u}_j}_\infty}.
		\end{aligned}
		\end{equation}
		
		Following the condition in Eq.~\eqref{eq:midCondi2}, we get:
		\begin{equation}
		\begin{aligned}
		&\sum_{j = \floor{N(1-\alpha)}}^{N-1} <\mathbf{S}^{(i)},\mathbf{u}_j> +\sum_{k=1}^{N} \mathbf{n}(k)\mathbf{u}_j(k)  \\
        &<-Eng(\hat{\mathbf{S}}^{(n)}) C_e .
        		\end{aligned}
		\end{equation}
        Therefore,
        \begin{equation}
		n_f< \frac{\sum_{j = \floor{N(1-\alpha)}}^{N-1} <\mathbf{S}^{(i)},\mathbf{u}_j> + C_e Eng(\hat{\mathbf{S}}^{(n)}) }{-2A(1+ \floor{N\alpha})\max_{j \in \big[\floor{N(1-\alpha)}, N-1\big]}\norm{\mathbf{u}_j}_\infty}.
		\end{equation}
    \end{proof}
 A similar bound can be derived for the LECR metric. 
% 	\begin{corollary}
% 			For the energy concentration low metric, given an $\epsilon \%$-prediction interval for random failures $[C_s,C_e]$, if $LECR_{\gamma} (\hat{\mathbf{S}^{(i)}})=C_i$, then the following bounds guarantee the correct detection upon existence of an epidemic:

% 		\textbf{Case 1: If $C_i<C_s$:}
% 		\begin{equation}
% 		\begin{aligned}
% 		&\resizebox{.285\vsize}{!}{$n_{f} < \Bigg[
% 		\frac{ C_s \widetilde{Eng}(\hat{\mathbf{S}}^{(n)}) -\sum_{j = 0}^{\ceil{N \gamma}-1} \big|\widetilde{\hat{\mathbf{S}}^{(i)} (\lambda_j)}\big|}{2 \ceil{N\gamma}\max_{j \in \big[0, \ceil{N \gamma}\big]} \norm{\mathbf{u}_j}_\infty}\Bigg]^+ .$}
% 		\end{aligned}
% 		\end{equation}
		
% 		\textbf{Case 2: If $C_i>C_e$:}
% 		\begin{equation}
% 		\begin{aligned}
% 		&\resizebox{.312\vsize}{!}{$n_{f} < \Bigg[\min \Bigg\{ \frac{\sum_{j = 0}^{\ceil{N \gamma}-1} \widetilde{\hat{\mathbf{S}}^{(i)} (\lambda_j)} + C_e \widetilde{Eng}(\hat{\mathbf{S}}^{(n)}) }{-2 \ceil{N\gamma}\max_{j \in \big[0, \ceil{N \gamma}\big]}\norm{\mathbf{u}_j}_\infty},$}\\
% 		&\resizebox{.235\vsize}{!}{$\frac{\sum_{j = 0}^{\ceil{N \gamma}-1} \widetilde{\hat{\mathbf{S}}^{(i)} (\lambda_j)} - C_e\widetilde{Eng}(\hat{\mathbf{S}}^{(n)}) }{2 \ceil{N\gamma}\max_{j \in \big[0, \ceil{N \gamma}\big]}\norm{\mathbf{u}_j}_\infty}\Bigg\}\Bigg]^+ .$}
% 		\end{aligned}
% 		\end{equation}
% 	\end{corollary}		
% 	\begin{proof}
% 		The proof is similar to the proof of Theorem~\ref{th:HECRerror}.
% 	\end{proof}
	\begin{theorem}\label{th:smbound}
    For the smoothness metric, given an $\epsilon \%$-prediction interval for random failures $[C_s,C_e]$, if $SM(\mathbf{S}^{(i)})=C_i$, then the following bounds guarantee the correct detection upon existence of an epidemic:
		\begin{align}\label{eq:end1}
		&n_f< \left[\frac{C_i-C_e}{2 A D^{\frac{3}{2}}  \sqrt{\Delta} + 2 A \sqrt{\Delta D}}\right]^+ \;\;\;\textrm{        if     } C_i>C_e,\\
		&n_f< \left[\frac{C_s- C_i}{2  A D^{\frac{3}{2}}  \sqrt{\Delta} + 2A \sqrt{\Delta D}}\right]^+ \;\;\; \textrm{        if      } C_i<C_s.
		\end{align}
		where $D$ is the maximum degree of the nodes and $\Delta$ is the maximum distance of the edges of the graph.
	\end{theorem}
    \begin{proof}
    Considering Eq.~\eqref{eq:smoothness}, the effect of adding $n_f$ faulty observations can be bounded as follows:
		\begin{flalign}\label{eq:smoothness2}
		SM (\mathbf{S}^{(n)})\leq SM (\mathbf{S}^{(i)}) + 2 n_f A D^{\frac{3}{2}}  \sqrt{\Delta} + 2n_f A \sqrt{\Delta D},
		\end{flalign}
		where the second term is associated with the neighbors of the faulty observations and the last term is associated with the faulty observations.
		Similarly, we can get:
		\begin{flalign}
		SM(\mathbf{S}^{(n)})\geq SM (\mathbf{S}^{(i)}) - 2 n_f A D^{\frac{3}{2}}  \sqrt{\Delta} - 2n_f A \sqrt{\Delta D}.
		\end{flalign}	
		If $C_i>C_e$, the correct detection is ensured upon having:
		\begin{equation}\label{eq:end1}
		\begin{aligned}
		2 n_f A D^{\frac{3}{2}}  \sqrt{\Delta} + 2n_f A \sqrt{\Delta D}< C_i-C_e .
		\end{aligned}
		\end{equation}
		If $C_i<C_S$, the correct detection is ensured upon having:
		\begin{equation}\label{eq:end2}
		\begin{aligned}
		2 n_f A D^{\frac{3}{2}}  \sqrt{\Delta} + 2n_f A \sqrt{\Delta D}< C_s- C_i .
		\end{aligned}
		\end{equation}
		The theorem result can be obtained from Eq. \eqref{eq:end1}, \eqref{eq:end2}.
    \end{proof}
    Note that $HECR$ and $LECR$ metrics are not sensitive to the value of $A$, which is not the case in the smoothness metric.

    \bibliographystyle{IEEEtran}
	
	\bibliography{IEEEabrv}

\begin{thebibliography}{10}
\providecommand{\url}[1]{#1}
\csname url@samestyle\endcsname
\providecommand{\newblock}{\relax}
\providecommand{\bibinfo}[2]{#2}
\providecommand{\BIBentrySTDinterwordspacing}{\spaceskip=0pt\relax}
\providecommand{\BIBentryALTinterwordstretchfactor}{4}
\providecommand{\BIBentryALTinterwordspacing}{\spaceskip=\fontdimen2\font plus
\BIBentryALTinterwordstretchfactor\fontdimen3\font minus
  \fontdimen4\font\relax}
\providecommand{\BIBforeignlanguage}[2]{{%
\expandafter\ifx\csname l@#1\endcsname\relax
\typeout{** WARNING: IEEEtran.bst: No hyphenation pattern has been}%
\typeout{** loaded for the language `#1'. Using the pattern for}%
\typeout{** the default language instead.}%
\else
\language=\csname l@#1\endcsname
\fi
#2}}
\providecommand{\BIBdecl}{\relax}
\BIBdecl

\bibitem{ref:ours}
S.~Hosseinalipour, J.~Wang, H.~Dai, and W.~Wang, ``Detection of infections
  using graph signal processing in heterogeneous networks,'' in \emph{Proc.
  IEEE Global Commun. Conf. (GLOBECOM)}, 2017, pp. 1--6.

\bibitem{knight2000iloveyou}
P.~Knight, ``Iloveyou: Viruses, paranoia, and the environment of risk,''
  \emph{The Sociological Rev.}, vol.~48, no.~S2, pp. 17--30, 2000.

\bibitem{wong2004studyMydoom}
C.~Wong, S.~Bielski, J.~M. McCune, and C.~Wang, ``A study of mass-mailing
  worms,'' in \emph{Proc. ACM workshop Rapid Malcode}, 2004, pp. 1--10.

\bibitem{06:Ayhan:TOIE}
B.~Ayhan, M.~Y. Chow, and M.~H. Song, ``Multiple discriminant analysis and
  neural-network-based monolith and partition fault-detection schemes for
  broken rotor bar in induction motors,'' \emph{IEEE Trans. Ind. Electron.},
  vol.~53, no.~4, pp. 1298--1308, June 2006.

\bibitem{08:Yin:TOKDE}
J.~Yin, Q.~Yang, and J.~J. Pan, ``Sensor-based abnormal human-activity
  detection,'' \emph{IEEE Trans. Knowl. Data Eng.}, vol.~20, no.~8, pp.
  1082--1090, Aug 2008.

\bibitem{10:Li:TOWC}
H.~Li and Z.~Han, ``Catch me if you can: An abnormality detection approach for
  collaborative spectrum sensing in cognitive radio networks,'' \emph{IEEE
  Trans. Wireless Commun.}, vol.~9, no.~11, pp. 3554--3565, November 2010.

\bibitem{Centola:2010:Science}
D.~Centola, ``The spread of behavior in an online social network experiment,''
  \emph{Science}, vol. 329, no. 5996, pp. 1194--1197, 2010.

\bibitem{ref:source}
Z.~Wang, W.~Dong, W.~Zhang, and C.~W. Tan, ``Rooting our rumor sources in
  online social networks: The value of diversity from multiple observations,''
  \emph{IEEE J. Sel. Topics Signal Process.}, vol.~9, no.~4, pp. 663--677, June
  2015.

\bibitem{ref:source2}
D.~Shah and T.~Zaman, ``Rumors in a network: Who's the culprit?'' \emph{IEEE
  Trans. Inf. Theory}, vol.~57, no.~8, pp. 5163--5181, Aug 2011.

\bibitem{Lloyd:2001:Science}
A.~L. Lloyd and R.~M. May, ``How viruses spread among computers and people,''
  \emph{Science}, vol. 292, no. 5520, pp. 1316--1317, 2001.

\bibitem{ref:restrain2}
Z.~He, Z.~Cai, J.~Yu, X.~Wang, Y.~Sun, and Y.~Li, ``Cost-efficient strategies
  for restraining rumor spreading in mobile social networks,'' \emph{IEEE
  Trans. Veh. Technol.}, vol.~66, no.~3, pp. 2789--2800, March 2017.

\bibitem{ref:restrain}
S.~Wen, J.~Jiang, Y.~Xiang, S.~Yu, W.~Zhou, and W.~Jia, ``To shut them up or to
  clarify: Restraining the spread of rumors in online social networks,''
  \emph{IEEE Trans. Parallel Distrib. Syst.}, vol.~25, no.~12, pp. 3306--3316,
  Dec 2014.

\bibitem{Milling:13:MobiHoc}
C.~Milling, C.~Caramanis, S.~Mannor, and S.~Shakkottai, ``Detecting epidemics
  using highly noisy data,'' in \emph{Proc. 14th ACM Int. Symp. Mobile Ad Hoc
  Netw. and Comput. (MobiHoc)}, 2013, pp. 177--186.

\bibitem{15:Milling:InfoCom}
------, ``Local detection of infections in heterogeneous networks,'' in
  \emph{IEEE Conf. Comput. Commun. (INFOCOM)}, April 2015, pp. 1517--1525.

\bibitem{ref:Internet}
A.-L. Barabási, R.~Albert, and H.~Jeong, ``Scale-free characteristics of
  random networks: the topology of the world-wide web,'' \emph{Physica A:
  Statistical Mechanics Appl.}, vol. 281, no.~1, pp. 69 -- 77, 2000.

\bibitem{Albert:02:RevPhys}
R.~Albert and A.-L. Barab\'asi, ``Statistical mechanics of complex networks,''
  \emph{Rev. Mod. Phys.}, vol.~74, pp. 47--97, Jan 2002.

\bibitem{Shuman:13:ISPM}
D.~I. Shuman, S.~K. Narang, P.~Frossard, A.~Ortega, and P.~Vandergheynst, ``The
  emerging field of signal processing on graphs: Extending high-dimensional
  data analysis to networks and other irregular domains,'' \emph{IEEE Signal
  Process. Mag.}, vol.~30, no.~3, pp. 83--98, May 2013.

\bibitem{Sandryhaila:14:TOSP}
A.~Sandryhaila and J.~M.~F. Moura, ``Discrete signal processing on graphs:
  Frequency analysis,'' \emph{IEEE Trans. Signal Process.}, vol.~62, no.~12,
  pp. 3042--3054, June 2014.

\bibitem{Sandryhaila:13:TOSP}
------, ``Discrete signal processing on graphs,'' \emph{IEEE Trans. Signal
  Process.}, vol.~61, no.~7, pp. 1644--1656, April 2013.

\bibitem{ref:GW1}
M.~Crovella and E.~Kolaczyk, ``Graph wavelets for spatial traffic analysis,''
  in \emph{Proc. 22nd IEEE Comput. Commun. Soc. (INFOCOM)}, vol.~3, March 2003,
  pp. 1848--1857.

\bibitem{ref:GW2}
D.~K. Hammond, P.~Vandergheynst, and R.~Gribonval, ``Wavelets on graphs via
  spectral graph theory,'' \emph{Appl. and Computational Harmonic Anal.},
  vol.~30, no.~2, pp. 129--150, 2011.

\bibitem{ref:wave1}
M.~K. Mihcak, I.~Kozintsev, K.~Ramchandran, and P.~Moulin, ``Low-complexity
  image denoising based on statistical modeling of wavelet coefficients,''
  \emph{IEEE Signal Process. Lett.}, vol.~6, no.~12, pp. 300--303, Dec 1999.

\bibitem{ref:wave2}
T.~F. Chan and J.~Shen, \emph{Image processing and analysis: variational, PDE,
  wavelet, and stochastic methods}.\hskip 1em plus 0.5em minus 0.4em\relax
  SIAM, 2005.

\bibitem{ref:naivebayes}
I.~Rish, ``An empirical study of the naive bayes classifier,'' in \emph{Proc.
  IJCAI workshop Empirical Methods Artificial Intell.}, vol.~3, no.~22.\hskip
  1em plus 0.5em minus 0.4em\relax IBM, 2001, pp. 41--46.

\bibitem{ref:randomForrest}
L.~Breiman, ``Random forests,'' \emph{Mach. learning}, vol.~45, no.~1, pp.
  5--32, 2001.

\bibitem{ref:eig1}
D.~Kempe and F.~McSherry, ``A decentralized algorithm for spectral analysis,''
  \emph{J. Comput. and Syst. Sci.}, vol.~74, no.~1, pp. 70--83, 2008.

\bibitem{ref:RFcomp}
G.~Louppe, ``Understanding random forests: From theory to practice,''
  \emph{arXiv preprint arXiv:1407.7502}, 2014.

\bibitem{ref:ERGraph}
P.~Er{\"o}ds and A.~R{\'e}nyi, ``On random graphs i,'' \emph{Publ. Math.
  Debrecen}, vol.~6, pp. 290--297, 1959.

\bibitem{ref:SFGraph}
A.-L. Barab{\'a}si and R.~Albert, ``Emergence of scaling in random networks,''
  \emph{science}, vol. 286, no. 5439, pp. 509--512, 1999.

\bibitem{snapnets}
J.~Leskovec and A.~Krevl, ``{SNAP Datasets}: {Stanford} large network dataset
  collection,'' \url{http://snap.stanford.edu/data}, Jun. 2014.

\bibitem{ref:emailnet}
H.~Yin, A.~R. Benson, J.~Leskovec, and D.~F. Gleich, ``Local higher-order graph
  clustering,'' in \emph{Proc. 23rd ACM SIGKDD Int. Conf. Knowledge Discovery
  and Data Mining}.\hskip 1em plus 0.5em minus 0.4em\relax ACM, 2017, pp.
  555--564.

\bibitem{SGWmex}
``The spectral graph wavelets toolbox,'' \url{https://wiki.epfl.ch/sgwt},
  accessed: 2018-05-05.

\bibitem{RGG}
J.~Dall and M.~Christensen, ``Random geometric graphs,'' \emph{Phys. Rev. E},
  vol.~66, p. 016121, Jul 2002.

\bibitem{RGGApp}
M.~Haenggi, J.~G. Andrews, F.~Baccelli, O.~Dousse, and M.~Franceschetti,
  ``Stochastic geometry and random graphs for the analysis and design of
  wireless networks,'' \emph{EEE J. Sel. Areas Commun.}, vol.~27, no.~7, 2009.

\bibitem{McAuley:2012}
J.~McAuley and J.~Leskovec, ``Learning to discover social circles in ego
  networks,'' in \emph{Proc. 25th Int. Conf. Neural Inform. Process. Syst.
  (NIPS)}, USA, 2012, pp. 539--547.

\end{thebibliography}

\end{document}